\documentclass[lettersize,journal]{IEEEtran}
\usepackage{amsmath,amsfonts}
\usepackage{array}
\usepackage[caption=false,font=normalsize,labelfont=sf,textfont=sf]{subfig}
\usepackage{textcomp}
\usepackage{stfloats}
\usepackage{url}
\usepackage{verbatim}
\usepackage{graphicx}
\usepackage{cite}
\usepackage[linesnumbered,ruled]{algorithm2e}
\usepackage{setspace}
\usepackage{amsthm}
\usepackage{booktabs}
\usepackage[acronym]{glossaries}
\newacronym{awgn}{AWGN}{additive white Gaussian noise}
\newacronym{ber}{BER}{bit error rate}
\newacronym{bs}{BS}{base station}
\newacronym{cnn}{CNN}{convolutional neural network}
\newacronym{gan}{GAN}{generative adversarial network}
\newacronym{csi}{CSI}{channel state information}
\newacronym{cdi}{CDI}{channel direction information}
\newacronym{dl}{DL}{deep learning}
\newacronym{tdd}{TDD}{time division duplex}
\newacronym{fdd}{FDD}{frequency division duplex}
\newacronym{los}{LoS}{line-of-sight}
\newacronym{mimo}{MIMO}{multiple-input multiple-output}
\newacronym{ml}{ML}{machine learning}
\newacronym{nlos}{NLoS}{non-line-of-sight}
\newacronym{nmse}{NMSE}{normalized mean squared error}
\newacronym{mse}{MSE}{mean squared error}
\newacronym{nn}{NN}{neural network}
\newacronym{ofdm}{OFDM}{orthogonal frequency division multiplexing}
\newacronym{pca}{PCA}{principal component analysis}
\newacronym{rnn}{RNN}{recurrent neural network}
\newacronym{snr}{SNR}{signal-to-noise ratio}
\newacronym{ue}{UE}{user equipment}
\newacronym{upa}{UPA}{uniform planar array}
\newacronym{cscg}{CSCG}{circularly symmetric complex Gaussian}
\newacronym{cdf}{CDF}{cumulative distribution function}
\newacronym{pdf}{PDF}{probability density function}
\newacronym{dft3}{DFT3}{3-dimensional discrete Fourier transform}
\newacronym{idft3}{IDFT3}{3-dimensional inverse discrete Fourier transform}
\newacronym{dft2}{DFT2}{2-dimensional discrete Fourier transform}
\newacronym{idft2}{IDFT2}{2-dimensional inverse discrete Fourier transform}
\newacronym{bpsk}{BPSK}{binary phase shift keying}
\newtheorem{proposition}{Proposition}

\begin{document}

\title{Machine Learning-Based CSI Feedback With Variable Length in FDD Massive MIMO}

\author{Matteo Nerini,
        Valentina Rizzello,~\IEEEmembership{Student Member,~IEEE},
        Michael Joham,~\IEEEmembership{Member,~IEEE},
        Wolfgang Utschick,~\IEEEmembership{Fellow,~IEEE},
        Bruno Clerckx,~\IEEEmembership{Fellow,~IEEE}
\thanks{M. Nerini and B. Clerckx are with the Department of Electrical and Electronic Engineering, Imperial College London, London, SW7 2AZ, U.K. (e-mail: \{m.nerini20, b.clerckx\}@imperial.ac.uk).}
\thanks{V. Rizzello, M. Joham and W. Utschick are with the Professur f\"ur Methoden der Signalverarbeitung, Technische Universit\"at M\"unchen, Munich, 80333, Germany. (e-mail: \{valentina.rizzello, joham, utschick\}@tum.de).}
\thanks{Manuscript received April 11, 2022; revised August 08, 2022.}}

\markboth{Journal of \LaTeX\ Class Files,~Vol.~14, No.~8, August~2021}%
{Nerini \MakeLowercase{\textit{et al.}}: Machine Learning-Based CSI Feedback With Variable Length in FDD Massive MIMO}

\IEEEpubid{0000--0000/00\$00.00~\copyright~2021 IEEE}

\maketitle

\begin{abstract}
To fully unlock the benefits of \gls{mimo} networks, downlink \gls{csi} is required at the \gls{bs}.
In \gls{fdd} systems, the \gls{csi} is acquired through a feedback signal from the \gls{ue}.
However, this may lead to an important overhead in \gls{fdd} massive \gls{mimo} systems.
Focusing on these systems, in this study, we propose a novel strategy to design the \gls{csi} feedback.
Our strategy allows to optimally design variable length feedback, that is promising compared to fixed feedback since users experience channel matrices differently sparse.
Specifically, \gls{pca} is used to compress the channel into a latent space with adaptive dimensionality.
To quantize this compressed channel, the feedback bits are smartly allocated to the latent space dimensions by minimizing the \gls{nmse} distortion.
Finally, the quantization codebook is determined with $k$-means clustering.
Numerical simulations show that our strategy improves the zero-forcing beamforming sum rate by 17\%, compared to CsiNetPro.
The number of model parameters is reduced by 23.4 times, thus causing a significantly smaller offloading overhead.
At the same time, \gls{pca} is characterized by a lightweight unsupervised training, requiring eight times fewer training samples than CsiNetPro.
\end{abstract}

\glsresetall


\section{Introduction}

\IEEEPARstart{A}{ccurate} knowledge of the wireless channel, or \gls{csi}, is critical to unlock the full potential benefits of \gls{mimo} systems.
In particular, \gls{csi} at the transmitter and at the receiver allow precoding and combining techniques, respectively, which can enhance the spectrum efficiency in \gls{mimo} wireless networks \cite{cle13}.
\gls{csi} at the transmitter can be easily acquired without feedback from the receiver in \gls{tdd} systems.
In these systems, channel reciprocity holds since the uplink and downlink channels share the same frequency band.
Conversely, \gls{csi} estimation at the transmitter is harder in \gls{fdd} systems, where the uplink-downlink channel reciprocity does not hold in general.
Thus, feedback messages from the \gls{ue} to the \gls{bs} are needed to gain downlink \gls{csi} at the \gls{bs} \cite{lov08}.
This may lead to significant overhead in massive \gls{mimo} systems, where a large number of antennas is employed.
To face this problem, various strategies have been investigated, such as downlink training techniques \cite{cho14}, distributed compressive \gls{csi} estimation for multi-user settings \cite{rao14}, and adaptive \gls{csi} feedback depending on the channel sparsity \cite{gao15}.
Additionally, novel transmission strategies robust with respect to reduced-dimensional \gls{csi} have been proposed \cite{adh13,xu14,dai16}.
\IEEEpubidadjcol

In recent years, \gls{dl} techniques have been also used for \gls{csi} estimation in \gls{fdd} systems, following two main research directions.
In the first, the uplink-to-downlink channel mapping is learned.
In this way, the feedback can be completely removed since the uplink channel knowledge at the \gls{bs} is sufficient to gain also the downlink channel knowledge.
In the second, the channel matrix is compressed with a \gls{dl} architecture to minimize the feedback.

The uplink-to-downlink mapping existence was firstly investigated in \cite{alr19}.
In \cite{alr19,arn19}, a fully connected \gls{nn} is used to map the uplink to the corresponding downlink channels.
The same task is solved more efficiently in \cite{yan19} with a sparse complex-valued \gls{nn}, and in \cite{don19,han20} with a \gls{cnn} treating the space-frequency channel matrix as an image.
Image processing techniques are exploited also in \cite{saf19,wan19}.
Finally, in \cite{riz20}, the channel is firstly compressed with an autoencoder to decrease the feature space dimensionality; then, the uplink-to-downlink mapping is achieved with random forests.
However, when more complex channel models are considered (e.g. accounting for a rich multipath and dynamic environment), the uplink-to-downlink mapping function becomes hard to approximate with \glspl{nn}.
The environment should be sampled with enough resolution to account for small-scale fading effects, which may require extensive sampling campaigns and huge datasets in practice.
For this reason, the vast majority of related studies used \gls{dl} techniques to reduce the feedback overhead, rather than completely remove it.

In \cite{guo21, riz21}, a deep autoencoder is used to embed the downlink channel matrix at the \gls{ue} into a latent space with reduced dimensionality.
In this way, only the embedding of the channel matrix needs to be fed back to the \gls{bs}.
Deep autoencoder-like architectures have been proposed for the same scope in \cite{wen18,liu19,liu20,mas21,cao21,sun21,jo21}.
The CsiNet architecture introduced in \cite{wen18} is extended in \cite{wan18} to exploit the temporal redundancy in a dynamic environment through recurrent layers, and further improved in \cite{guo20a,li20}.
In \cite{liu19}, the channel reconstruction accuracy is enhanced by exploiting also information from the uplink \gls{csi}, assumed available at the \gls{bs}.
In \cite{liu20,mas21}, a \gls{dl}-based framework is proposed to jointly learn \gls{csi} compression and quantization.
In \cite{cao21,sun21,jo21}, the \gls{csi} feedback is designed through lightweight \gls{dl} architectures, characterized by a reduced number of trainable parameters.
To improve the compression performance and reduce the model complexity, the novel concepts of network aggregation and layer binarization have been employed in \cite{lu22}.
Furthermore, the low rank structure of the millimeter wave channels has been exploited by the model-driven \gls{dl} architecture proposed in \cite{guo21b}.

Other related works employ \gls{dl} techniques either at the \gls{bs} or at the \gls{ue} to improve the channel reconstruction quality or the feedback design, respectively.
In \cite{sol19, uts21}, the feedback is generated at the \gls{ue} by sampling only specific entries of the downlink channel matrix.
Then, the original channel matrix is reconstructed at the \gls{bs} with image-restoration \glspl{cnn}.
Similarly, in \cite{lia20}, the channel is compressed at the \gls{ue} based on compressive sensing and then recovered through a \gls{dl} architecture at the \gls{bs}.
In \cite{tur21}, a \gls{nn} has been proposed to directly map noisy pilot observations to their optimal feedback index.
In \cite{che21}, \glspl{nn} are used to design an implicit \gls{csi} feedback mechanism in which the \gls{csi} is mapped into the recommended precoder.
Finally, \gls{dl} architectures have been applied to learn feedback messages with the objective of maximizing the downlink precoding performance in single-cell \cite{soh21}, and multi-cell \cite{guo20b} scenarios.

When applying \gls{dl} models for designing the \gls{csi} feedback, three problems arise.
First, very large training datasets are typically required: e.g., hundreds of thousands of downlink training samples are used in many of the aforementioned studies \cite{alr19,yan19,riz20,guo21,wen18,sun21,guo20a,lu22,lia20,che21,guo20b}.
This problem has been recently addressed by training the \gls{dl} architectures solely with uplink channel samples, assumed available at the \gls{bs} \cite{uts21,riz21,tur21,riz21b,riz22}.
These works exploit the conjecture that learning at the uplink frequency can be transferred at the downlink frequency with no further modification, as proposed for the first time in \cite{uts21} and validated in \cite{riz21b}.
Furthermore, this problem has been addressed through transfer learning in \cite{zen21}.
Second, to optimally generate variable length quantized feedback in massive \gls{mimo} systems, three problems need to be solved: \textit{i)} \gls{csi} compression into a latent space with dimensionality depending on the number of feedback bits, \textit{ii)} allocation of feedback bits to the latent space dimensions, and \textit{iii)} design of the quantization levels.
However, autoencoders compress the \gls{csi} into a latent space with fixed dimensionality, determined by the number of neurons present in their middle layer.
As a result, the latent space dimensionality cannot be adapted to the number of feedback bits, leading to performance degradation in two cases.
On the one hand, the performance significantly degrades when the number of feedback bits is too low compared to the latent space dimensionality, and it is not sufficient to properly quantize the too large latent space.
On the other hand, when a high number of feedback bits is considered, the performance does not improve beyond the upper bound caused by the fixed latent space dimensionality.
Two solutions to this problem are provided in \cite{guo20a,lia20}, where the authors propose a \gls{dl}-based method offering multiple compression ratios.
However, the number of available compression ratios is still limited to four in \cite{guo20a}, and the problem of optimally allocating the feedback bits to the compressed \gls{csi} dimensions is not investigated in neither study.
The problem of allocating the bits to the compressed \gls{csi} dimensions is relevant since they may carry different amount of information.
Nevertheless, this problem remains unexplored in recent literature proposing \gls{dl} solutions.
Third, each \gls{ue} must encode the \gls{csi} by using the \gls{dl} architecture trained at the \gls{bs}.
Thus, the trained architecture parameters need to be offloaded from the \gls{bs} to the \gls{ue} when a new \gls{ue} connects to the \gls{bs}, and after each training session.
This causes an additional fixed overhead in the transmission, which cannot be reduced.

To solve these three problems affecting \gls{dl} models, we propose an alternative approach avoiding the use of \gls{dl}.
Differently from previous literature, we compress the \gls{csi} with \gls{pca}, a classical \gls{ml} technique \cite{bis06}.
\gls{pca} offers a useful property that cannot be found in typical \gls{dl} architectures used for compression: the latent space dimensions explain different variances and can be ordered according to their importance.
This property allows us to introduce a novel and practical bit allocation technique that assigns the available feedback bits to the latent space dimensions.
Coupling \gls{pca} with our novel bit allocation, we show that our technique performs better, or approximately equal, than \gls{dl} architectures recently proposed for the same scope.
Besides the channel reconstruction quality improvement, our approach brings three fundamental benefits compared to related works using \gls{dl} solutions.
First, it is characterized by a lightweight training phase not involving costly iterative optimization algorithms and requiring a reduced number of training samples.
Second, our technique adapts the latent space dimensionality to the number of feedback bits.
Thus, variable length feedbacks can be optimally generated with a unique trained model. 
Third, our technique can be implemented with a variable number of model parameters.
This enables the network operator to trade the reconstruction quality and the communication overhead due to the model parameter offloading in an adaptive manner.
These benefits allow our approach to efficiently design variable length \gls{csi} feedback, whose advantages are promising compared to fixed feedback \cite{gao15,cle08}.
It has been shown that the \gls{csi} reconstruction can be greatly improved by adjusting the feedback overhead according to the sparsity level of the channels \cite{gao15}.
Also in a multi-user setting, the sum rate improves when the total feedback bits are smartly distributed among the users, employing a feedback with adaptive length \cite{cle08}.
The contributions of this paper are summarized as follows.

\textit{First}, we propose a novel \gls{csi} feedback strategy based on \gls{pca} and $k$-means clustering.
In this strategy, \gls{pca} is used to compress the channel matrix into a latent space with adaptive dimensionality, allowing to optimally design the feedback with variable length.
To quantize this compressed channel, the feedback bits are smartly allocated to the principal components in order to minimize a properly defined distortion function.
Finally, on each principal component, the quantization levels are determined with vector quantization.
Specifically, we employ $k$-means clustering since it is the fixed-rate quantization strategy that minimizes the \gls{mse} by construction.

\textit{Second}, we provide theoretical justifications to prove the optimality of our bit allocation to the principal components.
In our approach, the bits are allocated to the principal components by minimizing the \gls{nmse} distortion introduced by the compression and quantization operations.
Despite optimal bit allocation is a popular concept in wireless communications \cite{cle08}, the optimal feedback bits allocation to the \gls{csi} principal components has never been investigated before.
A practical iterative algorithm is proposed to this scope, whose optimality is theoretically guaranteed.

\textit{Third}, we propose an offloading overhead-aware \gls{csi} feedback by improving our \gls{pca}-based strategy with two modifications.
The first reduces the number of offloaded \gls{pca} parameters, while the second reduces the number of offloaded codebook parameters.
With these two modifications, the number of offloaded parameters can be adapted.
Thus, the network operator can select the optimal trade-off between offloading overhead and \gls{csi} reconstruction quality.
To the best of our knowledge, an adaptive offloading overhead has never been considered in previous literature employing \gls{dl} solutions.
Results show that these two modifications only slightly impact the reconstruction performance, while they significantly reduce the number of offloaded parameters.

\textit{Fourth}, we compare the performance of our strategy with two state-of-the-art \gls{dl} autoencoders proposed for the same scope.
As a benchmark, we consider CsiNetPro \cite{li20}, the improved version of the popular CsiNet \cite{wen18}, and the autoencoder more recently proposed in \cite{riz21} whose training is based on uplink data.
Our strategy is characterized by a lightweight training phase, requiring significantly fewer training samples.
The number of offloaded model parameters can be reduced with approximately no \gls{csi} reconstruction quality loss.
Consequently, our strategy causes less offloading overhead than the considered reference autoencoders.
At the same time, our strategy achieves better or approximately the same \gls{csi} reconstruction quality as these autoencoders.

\textit{Organization}: In Section \ref{sec:system-model}, we define the system model and the problem formulation.
In Section~\ref{sec:feedback}, we present our novel \gls{csi} feedback strategy based on \gls{pca} and $k$-means clustering.
In Section~\ref{sec:overhead}, we improve our \gls{csi} feedback strategy by proposing two modifications that allow to significantly decrease the number of offloaded parameters.
In Section~\ref{sec:results}, we assess the obtained performance through numerical simulations.
Finally, Section~\ref{sec:conclusion} contains the concluding remarks.
For reproducible research, the simulation code is available at \url{https://github.com/matteonerini/ml-based-csi-feedback}.

\textit{Notation}: Vectors and matrices are denoted with bold lower and bold upper letters, respectively.
Scalars are represented with letters not in bold font.
$\vert a\vert$, $\Vert\mathbf{a}\Vert$, and $\Vert\mathbf{A}\Vert_{F}$ refer to the absolute value of a complex scalar $a$, $l_{2}$-norm of a vector $\mathrm{\mathbf{a}}$, and Frobenius norm of a matrix $\mathrm{\mathbf{A}}$, respectively.
$\left[\mathrm{\mathbf{a}}\right]_{i}$ denotes the $i$-th element of vector $\mathbf{a}$.
$\mathrm{\mathbf{A}}^{T}$ and $\mathrm{\mathbf{A}}^{H}$ refer to the transpose and conjugate transpose of a matrix $\mathrm{\mathbf{A}}$, respectively.
$\mathbb{N}$, $\mathbb{R}$, and $\mathbb{C}$ denote natural, real, and complex number sets, respectively.
$\mathbb{N^*}$ denotes the positive natural number set.
$\mathbf{0}$ and $\mathbf{I}$ refer to an all-zero matrix and an identity matrix, respectively.
$\mathcal{CN}\left(0,\sigma^2\right)$ denotes the distribution of a \gls{cscg} random variable whose real and imaginary parts are independent normally distributed with mean zero and variance $\sigma^2/2$.
$\mathcal{CN}\left(\mathbf{0},\mathbf{R}\right)$ denotes the distribution of a \gls{cscg} random vector with mean vector $\mathbf{0}$ and covariance matrix $\mathbf{R}$.
Finally, $\mathrm{diag}\left(a_{1},\dots,a_{N}\right)$ refers to a diagonal matrix whose diagonal elements are $a_{1},\dots,a_{N}$.

\section{System Model}
\label{sec:system-model}

Let us consider a multi-user \gls{mimo} system in which a \gls{bs} equipped with $N_{A}$ antennas serves $K$ single-antenna users.
Assuming an \gls{fdd} communication system, we denote with $f_{\text{UL}}$ and $f_{\text{DL}}$ the uplink and downlink center frequencies, respectively.
In both frequencies, the channel bandwidth $W$ is divided into $N_{C}$ \gls{ofdm} subcarriers.

We denote as $x_{\text{UL},q,n_{C}}\in\mathbb{C}$ the uplink signal transmitted by the $q$-th user on the $n_{C}$-th uplink subcarrier, and as $\mathbf{y}_{\text{UL},n_{C}}\in\mathbb{C}^{N_{A}\times 1}$ the uplink signal received by the \gls{bs} on the $n_{C}$-th uplink subcarrier.
Thus, we have
\begin{equation}
\mathbf{y}_{\text{UL},n_{C}} = \sum_{q=1}^{K} \mathbf{h}_{\text{UL},q,n_{C}}x_{\text{UL},q,n_{C}}+\mathbf{n}_{\text{UL},n_{C}},
\end{equation}
where $\mathbf{h}_{\text{UL},q,n_{C}}\in\mathbb{C}^{N_{A}\times 1}$ is the uplink channel vector seen by the $q$-th user on the $n_{C}$-th uplink subcarrier, and $\mathbf{n}_{\text{UL},n_{C}}$ is the \gls{awgn} at the $n_{C}$-th uplink subcarrier.
Similarly, we denote as $\mathbf{x}_{\text{DL},n_{C}}\in\mathbb{C}^{N_{A}\times 1}$ the downlink signal transmitted by the \gls{bs} on the $n_{C}$-th downlink subcarrier, and as $y_{\text{DL},q,n_{C}}$ the downlink signal received by the $q$-th user on the $n_{C}$-th downlink subcarrier.
Thus, we can write
\begin{equation}
y_{\text{DL},q,n_{C}} = \mathbf{h}_{\text{DL},q,n_{C}}\mathbf{x}_{\text{DL},n_{C}}+n_{\text{DL},q,n_{C}},
\end{equation}
where $\mathbf{h}_{\text{DL},q,n_{C}}\in\mathbb{C}^{1\times N_{A}}$ is the downlink channel vector seen by the $q$-th user on the $n_{C}$-th downlink subcarrier, and $n_{\text{DL},q,n_{C}}$ is the \gls{awgn} at the $n_{C}$-th downlink subcarrier.

Let us define $\mathbf{H}_{\text{UL},q}=\left[\mathbf{h}_{\text{UL},q,1},\mathbf{h}_{\text{UL},q,2},\ldots,\mathbf{h}_{\text{UL},q,N_{C}}\right]\in\mathbb{C}^{N_{A}\times N_{C}}$ as the uplink channel matrix and $\mathbf{H}_{\text{DL},q}=\left[\mathbf{h}_{\text{DL},q,1}^{T},\mathbf{h}_{\text{DL},q,2}^{T},\ldots,\mathbf{h}_{\text{DL},q,N_{C}}^{T}\right]\in\mathbb{C}^{N_{A}\times N_{C}}$ as the downlink channel matrix of the $q$-th user.
For simplicity of notation, we drop the index $q$ in the rest of the paper.
We assume that noisy versions of the uplink channel matrices are available at the \gls{bs}, denoted as
\begin{equation}
\widetilde{\mathbf{H}}_{\text{UL}}=\mathbf{H}_{\text{UL}}+\mathbf{N}_{\text{UL}},
\end{equation}
where the entries of $\mathbf{N}_{\text{UL}}\in\mathbb{C}^{N_{A}\times N_{C}}$ are \gls{cscg} distributed such that $\mathrm{E}[\left\Vert\mathbf{H}_{\text{UL}}\right\Vert_F^2/\left\Vert\mathbf{N}_{\text{UL}}\right\Vert_F^2]={SNR}_{\text{UL}}$.
Similarly, a noisy version of the downlink channel matrix is available at each \gls{ue}, which is given by
\begin{equation}
\widetilde{\mathbf{H}}_{\text{DL}}=\mathbf{H}_{\text{DL}}+\mathbf{N}_{\text{DL}},
\end{equation}
where $\mathrm{E}[\left\Vert\mathbf{H}_{\text{DL}}\right\Vert_F^2/\left\Vert\mathbf{N}_{\text{DL}}\right\Vert_F^2]={SNR}_{\text{DL}}$.
In this system model, our goal is to reconstruct the downlink \gls{csi} $\mathbf{H}_{\text{DL}}$ at the \gls{bs} from limited feedback sent by the \gls{ue}.
To this end, we compress and quantize the \gls{csi} using two \gls{ml} techniques: \gls{pca} and $k$-means clustering.

\section{PCA and $k$-means Clustering-Based CSI Feedback Design}
\label{sec:feedback}

In this section, we describe our novel channel estimation process, which consists of two stages.
In the first stage, called \emph{offline learning}, the compression strategy is learned at the \gls{bs}, based on a reduced training dataset.
This stage has three objectives: learning the principal components to compress the \gls{csi}, allocating the feedback bits to the first principal components, and learning the quantization levels through $k$-means clustering.
In the second stage, called \emph{online development}, the \gls{csi} is compressed and quantized at the \gls{ue}, sent to the \gls{bs}, and here reconstructed, as graphically represented in Fig.~\ref{fig:pca}.

\begin{figure*}[t]
    \centering
    \includegraphics[width=0.86\textwidth]{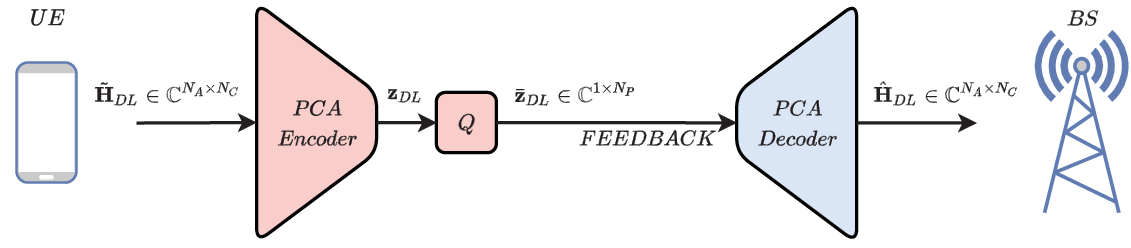}
    \caption{Downlink channel estimation at the BS with quantized feedback from the UE.}
    \label{fig:pca}
\end{figure*}

\subsection{Offline Learning}

To learn the principal components, we propose an offline training stage carried out at the \gls{bs}.
To this end, we consider a training set solely composed of noisy uplink channels, assumed available at the \gls{bs}.
This approach is based on the conjecture that learning at the uplink frequency can be transferred at the downlink frequency with no further modification \cite{uts21,riz21b}.
In this way, downlink \gls{csi} samples, which are hard to collect at the \gls{bs}, are not required in the training stage.
Let us denote with $N_{\text{train}}$ the number of training samples, and with $\widetilde{\mathbf{H}}_{\text{train}}\in\mathbb{C}^{N_{\text{train}}\times N_{A}N_{C}}$ the training set, whose rows are vectorized noisy uplink channels.
Furthermore, we denote with $\boldsymbol{\mu}_{\text{train}}$ the column-wise mean of $\widetilde{\mathbf{H}}_{\text{train}}$, and with $\mathbf{H}_{\text{train}}$ the centered training set obtained by subtracting $\boldsymbol{\mu}_{\text{train}}$ to each row of $\widetilde{\mathbf{H}}_{\text{train}}$.
Thus, the sample covariance matrix of the training set can be obtained as
\begin{equation}
\mathbf{R} =\frac{{\mathbf{H}_{\text{train}}}^{H}\mathbf{H}_{\text{train}}}{N_{\text{train}}-1}.
\end{equation}
With eigenvalue decomposition, $\mathbf{R}$ can be factorized as $\mathbf{R}= \mathbf{V}\mathbf{\Lambda}{\mathbf{V}}^{H}$, where $\mathbf{V}$ is a unitary matrix whose columns are the principal components, and $\mathbf{\Lambda}=\mathrm{diag}\left(\sigma_{1}^{2},\dots,\sigma_{N_{A}N_{C}}^{2}\right)$ is a diagonal matrix containing the $N_{A}N_{C}$ eigenvalues.
Let us assume that the eigenvalues in $\mathbf{\Lambda}$ have been sorted in descending order, and that their respective eigenvectors in $\mathbf{V}$ (i.e., the principal components) have been ordered accordingly.
Thus, we can select the conversion matrix $\mathbf{V}_{N_P}\in\mathbb{C}^{N_{A}N_{C}\times N_{P}}$, containing the first $N_{P}$ columns of $\mathbf{V}$, to project the data onto a latent space with dimensionality $N_{P}$.

After the dimensionality reduction operation, the compressed \gls{csi} needs to be discretized and represented with a limited number of bits $B$.
According to a naive discretization approach, we could fix the latent space dimensionality $N_{P}$ arbitrarily, and quantize the compressed \gls{csi} by allocating $B/N_{P}$ bits to each principal component.
However, this approach would be suboptimal for two reasons.
Firstly, $N_{P}$ should be defined according to the number of feedback bits $B$.
A higher $B$ should allow a larger latent space dimensionality $N_{P}$, while, vice versa, a lower $B$ should reduce $N_{P}$.
Secondly, the $N_{P}$ principal components have decreasing importance, defined as the variance explained by each of them.
Thus, the optimal discretization approach is expected to allocate more bits to more important principal components, and fewer bits to the less important ones.
To solve these two problems, we propose to allocate the $B$ bits to the principal components in a way that minimizes the distortion introduced by the compression and quantization operations.
As a result of this bit allocation process, the number of bits allocated to each principal component is described by the vector $\mathbf{b}=\left[b_1,\dots,b_{N_{A}N_{C}}\right]\in\mathbb{N}^{N_AN_C}$, where $b_{n}$ is the number of bits allocated to the $n$-th principal component.
Thus, we have $\sum_{n=1}^{N_{A}N_{C}} b_{n}=B$ and $N_{P}$ is given by the index of the last non-zero element in $\mathbf{b}$.
The proposed bit allocation process is described in detail in Subsection~\ref{subsec:bit-allocation}.

After determining $N_{P}$ and the bit allocation to the first $N_{P}$ principal components, the quantization levels need to be specified.
To minimize the quantization distortion, measured in terms of \gls{mse}, we define different quantization levels for each principal component, based on $k$-means clustering \cite{bis06}.
To this end, we project the training set $\mathbf{H}_{\text{train}}$ onto the first $N_{P}$ principal components, obtaining $\mathbf{Z}_{\text{train}}=\mathbf{H}_{\text{train}}\mathbf{V}_{N_{P}}\in\mathbb{C}^{N_{\text{train}}\times N_{P}}$.
In this way, the rows of $\mathbf{Z}_{\text{train}}$ are the embeddings of the training set elements into the latent space.
Then, $k$-means clustering is applied independently to each column of $\mathbf{Z}_{\text{train}}$, where for the $n$-th column we have $k=2^{b_{n}}$.
Since each column of $\mathbf{Z}_{\text{train}}$ is a complex vector, each $k$-means clustering problem is considered as a vector quantization problem in the feature space $\mathbb{R}^{2}$.
In this problem, the two dimensions are given by the real and imaginary parts of the $\mathbf{Z}_{\text{train}}$ elements.
Finally, the cluster centers of the $n$-th $k$-means clustering give the quantization levels for the $n$-th latent space dimension.

\subsection{Online Development}

After the \emph{offline learning} stage, the first $N_{P}$ principal components, the vector $\boldsymbol{\mu}_{\text{train}}$, and the codebook of the learned quantization levels are offloaded to the \gls{ue}.
This allows the online feedback process, as represented in Fig.~\ref{fig:pca}.
The \gls{ue} firstly vectorizes the noisy downlink channel matrix $\widetilde{\mathbf{H}}_{\text{DL}}$ into $\tilde{\mathbf{h}}_{\text{DL}}\in\mathbb{C}^{1\times N_{A}N_{C}}$.
Then, it compresses $\mathbf{h}_{\text{DL}}=\tilde{\mathbf{h}}_{\text{DL}}-\boldsymbol{\mu}_{\text{train}}$ by considering only the first $N_{P}$ principal components, obtaining $\mathbf{z}_{\text{DL}}=\mathbf{h}_{\text{DL}}\mathbf{V}_{N_P}\in\mathbb{C}^{1\times N_{P}}$ .
Before the feedback transmission, the compressed \gls{csi} $\mathbf{z}_{\text{DL}}$ is quantized according to the learned codebook.
Therefore, $\bar{\mathbf{z}}_{\text{DL}}$ is obtained by mapping the $\mathbf{z}_{\text{DL}}$ entries to their nearest cluster centers.
In this way, the feedback is fully described by $B$ bits, which are sent from the \gls{ue} to the \gls{bs}.
Finally, the downlink channel estimation process is concluded at the \gls{bs}.
Here, the decoder firstly calculates $\hat{\mathbf{h}}_{\text{DL}}=\bar{\mathbf{z}}_{\text{DL}}\mathbf{V}_{N_P}^{H}\in\mathbb{C}^{1\times N_{A}N_{C}}$; and secondly, it reshapes $\hat{\mathbf{h}}_{\text{DL}}+\boldsymbol{\mu}_{\text{train}}$ into $\widehat{\mathbf{H}}_{\text{DL}}$, which is the downlink channel matrix estimate.

\subsection{Bit Allocation to the Principal Components}
\label{subsec:bit-allocation}

Within the the \emph{offline learning} stage, the optimal bit allocation to the principal components is learned.
As anticipated, this is achieved by minimizing the distortion introduced by the compression and quantization operations.
To measure such a distortion, we define a distortion function $D_f$ given by the \gls{nmse} between the reconstructed vectorized channel
$\hat{\mathbf{h}}_{\text{DL}}$ and the actual vectorized channel available at the UE $\mathbf{h}_{\text{DL}}$
\begin{equation}
D_f = \mathbb{E}\left[\frac{\Vert\mathbf{h}_{\text{DL}}-\hat{\mathbf{h}}_{\text{DL}}\Vert^2}{\Vert\mathbf{h}_{\text{DL}}\Vert^2}\right],
\end{equation}
where the average is over all the channel realizations in the training set.
Without loss of generality, we can write $\mathbf{h}_{\text{DL}}=\mathbf{g}_{\text{DL}}\mathbf{V}^H$, where $\mathbf{g}_{\text{DL}}=\mathbf{h}_{\text{DL}}\mathbf{V}$ is the mapping of $\mathbf{h}_{\text{DL}}$ into the entire principal component space, and $\hat{\mathbf{h}}_{\text{DL}}=\bar{\mathbf{g}}_{\text{DL}}\mathbf{V}^H$, where $\bar{\mathbf{g}}_{\text{DL}}$ is the quantized version of $\mathbf{g}_{\text{DL}}$.
Thus, the distortion function $D_f$ can be rewritten as
\begin{align}
D_f
& = \mathbb{E}\left[\frac{\Vert\mathbf{g}_{\text{DL}}\mathbf{V}^H-\bar{\mathbf{g}}_{\text{DL}}\mathbf{V}^H\Vert^2}{\Vert\mathbf{h}_{\text{DL}}\Vert^2}\right]\\
& = \mathbb{E}\left[\frac{\Vert\left(\mathbf{g}_{\text{DL}}-\bar{\mathbf{g}}_{\text{DL}}\right)\mathbf{V}^H\Vert^2}{\Vert\mathbf{h}_{\text{DL}}\Vert^2}\right]\\
& \leq \mathbb{E}\left[\frac{\Vert\mathbf{g}_{\text{DL}}-\bar{\mathbf{g}}_{\text{DL}}\Vert^2\Vert\mathbf{V}^H\Vert_F^2}{\Vert\mathbf{h}_{\text{DL}}\Vert^2}\right]=\bar{D}_{f},
\end{align}
where the inequality follows from the Cauchy-Schwarz inequality.
We now assume that the training set has been normalized such that for every channel we have $\Vert\mathbf{h}_{\text{DL}}\Vert^2=N_{A}N_{C}$.
This has no impact on the multi-user precoding performance since the precoder design only requires the \gls{cdi}.
Thus, noting that $\Vert\mathbf{V}^H\Vert_F^2$ and $\Vert\mathbf{h}_{\text{DL}}\Vert^2$ are constant, the upper bound on the distortion function $\bar{D}_{f}$ can be simplified as
\begin{align}
\bar{D}_{f}
& = \mathbb{E}\left[\Vert\mathbf{g}_{\text{DL}}-\bar{\mathbf{g}}_{\text{DL}}\Vert^2\right]\\
& = \mathbb{E}\left[\sum_{n=1}^{N_{A}N_{C}} \vert\left[\mathbf{g}_{\text{DL}}\right]_{n}-\left[\bar{\mathbf{g}}_{\text{DL}}\right]_{n}\vert^2\right]\\
& = \sum_{n=1}^{N_{A}N_{C}} \mathbb{E}\left[\vert\left[\mathbf{g}_{\text{DL}}\right]_{n}-\left[\bar{\mathbf{g}}_{\text{DL}}\right]_{n}\vert^2\right].
\end{align}

We notice that the term $\mathbb{E}\left[\vert\left[\mathbf{g}_{\text{DL}}\right]_{n}-\left[\bar{\mathbf{g}}_{\text{DL}}\right]_{n}\vert^2\right]$ is the MSE due to the quantization of the $n$-th principal component.
The MSE is a common metric used to measure the distortion in lossy compression.
Thus, to simplify the notation, we introduce the vector $\mathbf{d}=\left[d_1,\dots,d_{N_{A}N_{C}}\right]$, where $d_{n}=\mathbb{E}\left[\vert\left[\mathbf{g}_{\text{DL}}\right]_{n}-\left[\bar{\mathbf{g}}_{\text{DL}}\right]_{n}\vert^2\right]$.
Consequently, our objective becomes to find the bit allocation $\mathbf{b}$ such that $\bar{D}_{f}=\sum_{n=1}^{N_{A}N_{C}}d_{n}$ is minimized.
Here, the implicit dependence of $\bar{D}_{f}$ on $\mathbf{b}$ lies in the fact that the distortion $d_{n}$ depends on the number of bits $b_{n}$ used to quantize the $n$-th principal component, and on the quantization strategy employed.
Since the fixed-rate quantization strategy that minimizes $d_{n}$ is $k$-means clustering, we consider $d_{n}$ as the minimum distortion obtainable by $k$-means clustering with $k=2^{b_{n}}$.
In the following, we write $d_{n}=d_{n}\left(b_{n}\right)$ to highlight that $d_{n}$ is a function purely dependent on $b_{n}$.

To design an algorithm able to find the optimal bit allocation $\mathbf{b}$, we now introduce two useful propositions.
Let us assume that the channels are Rayleigh distributed, with covariance matrix $\mathbf{R}$.
Thus, we have that the rows of $\mathbf{H}_{\text{train}}$ are distributed as $\mathcal{CN}\left(\mathbf{0},\mathbf{R}\right)$.
Recalling that $\mathbf{R}= \mathbf{V}\mathbf{\Lambda}{\mathbf{V}}^{H}$, the rows of $\mathbf{G}_{\text{train}}=\mathbf{H}_{\text{train}}\mathbf{V}$ are distributed as $\mathcal{CN}\left(\mathbf{0},\mathbf{\Lambda}\right)$ \cite{tse05}.
In other words, since $\mathbf{G}_{\text{train}}$ is the training set projection onto the whole principal components space, we obtain that on every principal component the training set entries are \gls{cscg} distributed.
More precisely, each entry of the $n$-th column of $\mathbf{G}_{\text{train}}$ is distributed as $\mathcal{CN}\left(0,\sigma_{n}^{2}\right)$, meaning that on the $n$-th principal component the training set entries are distributed as $\mathcal{CN}\left(0,\sigma_{n}^{2}\right)$.
This property is used to derive the following two propositions.

\begin{proposition}
If the optimal allocation of $B$ bits is $\mathbf{b}=\left[b_1,\dots,b_{N_{A}N_{C}}\right]$, the optimal allocation of $B+1$ bits is $\mathbf{b}^\prime=\left[b_1,\dots,b_{m}+1,\dots,b_{N_{A}N_{C}}\right]$ for some index $m\in\left[1,N_AN_C\right]$.
\label{pro:1}
\end{proposition}
\begin{proof}
Please refer to Appendix~A.
\end{proof}
In other words, the optimal allocation of $B+1$ bits can be obtained from the optimal allocation of $B$ bits by adding one bit to the $m$-th principal component, for some $m\in\left[1,N_AN_C\right]$.
This allows us to recursively find the optimal bit allocation.
\begin{proposition}
If the optimal allocation of $B$ bits is $\mathbf{b}=\left[b_1,\dots,b_{N_{A}N_{C}}\right]$, we have $b_{n} \leq b_{n-1}\;\forall n\in\left[2,N_AN_C\right]$.
\label{pro:2}
\end{proposition}
\begin{proof}
Please refer to Appendix~B.
\end{proof}
Note that the $\mathbf{b}$ elements are non increasing since the principal components are ordered according to their importance.
This agrees with our intuition that more feedback bits should be allocated to more important principal components, and fewer bits to the less important ones.

\begin{algorithm}[t]
\KwIn{$B$, $\mathbf{G}_{\text{train}}$}
\KwOut{$\mathbf{b}$}
%
$\mathbf{b}\leftarrow\left[0,\dots,0\right]\in\mathbb{N}^{N_AN_C}$\;
\For{$i\leftarrow 1$ \KwTo $B$}{
    $\mathbf{\Delta d}\leftarrow\left[0,\dots,0\right]\in\mathbb{R}^{N_AN_C}$\;
    \For{$n\leftarrow 1$ \KwTo $N_{A}N_{C}$}{ 
        \uIf{$n==1$}{
            $\Delta d_{n}\leftarrow d_{n}(b_{n})-d_{n}(b_{n}+1)$\;
        }
        \ElseIf{$b_{n}<b_{n-1}$}{
            $\Delta d_{n}\leftarrow d_{n}(b_{n})-d_{n}(b_{n}+1)$\;
        }
    }
    $\left[\sim,m\right]\leftarrow \text{max}(\mathbf{\Delta d})$\;
    $b_{m}\leftarrow b_{m}+1$\;
}
\KwRet{$\mathbf{b}$}
\caption{Bit allocation to the principal components.}
\label{alg:bit-allocation}
\end{algorithm}

We now use these two propositions to find the optimal bit allocation $\mathbf{b}$ trough Alg.~\ref{alg:bit-allocation}.
The objective of Alg.~\ref{alg:bit-allocation} is to return the bit allocation vector $\mathbf{b}$ that minimizes the distortion function $\bar{D}_{f}=\sum_{n=1}^{N_{A}N_{C}}d_{n}$.
To do so, we propose an iterative approach in which one bit at a time is added in $\mathbf{b}$, until all $B$ bits are allocated, i.e., until $\sum_{n=1}^{N_{A}N_{C}} b_{n}=B$.
Note that this iterative approach is able to lead to the optimal bit allocation because of Proposition~\ref{pro:1}.
At the $i$-th iteration, the index of the $\mathbf{b}$ element receiving the $i$-th bit could be found with exhaustive search.
More precisely, for all the $N_AN_C$ principal components, the decrease in distortion caused by the additional bit could be computed and stored in the vector $\mathbf{\Delta d}=\left[\Delta d_1,\dots,\Delta d_{N_AN_C}\right]\in\mathbb{R}^{N_AN_C}$, where $\Delta d_{n}=d_{n}(b_{n})-d_{n}(b_{n}+1)$.
Consequently, the $i$-th bit could be allocated to the principal component experiencing the highest decrease in distortion (i.e., the $m$-th principal component in Alg.~\ref{alg:bit-allocation}).
However, this exhaustive search would require to explore $N_AN_C$ possibilities for the allocation of each bit.
To significantly reduce the search space from $N_AN_C$ to only a few plausible principal components we resort to Proposition~\ref{pro:2}.
If at the $i$-th iteration we have $b_{n}=b_{n-1}$, the $n$-th principal component will certainly not receive the $i$-th bit.
Thus, the index $n$ can be excluded from the search.
As a consequence of Proposition~\ref{pro:2}, the number of plausible principal components that are eligible to receive the $i$-th bit becomes $b_1+1$ (i.e., a few units in practice).

The inputs of Alg.~\ref{alg:bit-allocation} are the number of feedback bits $B$, and the matrix $\mathbf{G}_{\text{train}}$, representing the projection of the training set $\mathbf{H}_{\text{train}}$ onto the whole principal components space.
Here, $\mathbf{G}_{\text{train}}$ is needed to compute $d_{n}\left(b_{n}\right)$ trough $k$-means clustering.\footnote{Note that if $N_{\text{train}}\leq N_{A}N_{C}$, only the first $N_{\text{train}}-1$ columns of $\mathbf{G}_{\text{train}}$ are not null.
Nevertheless, this does not affect our discussion if $N_{\text{train}}$ is sufficiently large, i.e., $N_{\text{train}}>N_{P}$, where $N_{P}$ is the last non-zero element in $\mathbf{b}$.}
Furthermore, $d_{n}\left(0\right)$ is the maximum distortion for the $n$-th principal component, obtained by quantizing it with 0 bits.
Thus, we have that $d_{n}\left(0\right)$ is the variance of the $n$-th column of $\mathbf{G}_{\text{train}}$, i.e., $d_{n}\left(0\right)=\sigma_{n}^{2}$.
Finally, the output of Alg.~\ref{alg:bit-allocation} is the vector $\mathbf{b}$, whose $n$-th element represents the number of bits allocated to the $n$-th principal component.
The latent space dimensionality $N_{P}$ is returned implicitly as the index of the last non-zero element in $\mathbf{b}$.
%
We remark that this bit allocation strategy is possible because of a specific property of \gls{pca}.
In \gls{pca}, once the matrix $\mathbf{V}$ is constructed upon a single training phase, all the latent space dimensionalities $N_{P} \in [1, \min\{N_{\text{train}}-1, N_{A}N_{C}\}]$ can be considered for compression.
Conversely, in autoencoders, this degree of freedom is not present since the latent space dimensionality is completely determined by the number of its middle layer neurons.

\subsection{Computational Complexity}
The computational load required by our \gls{csi} feedback strategy is here assessed. 
Firstly, the complexity needed to obtain the sample covariance matrix $\mathbf{R}$ and its eigenvalue decomposition is $\mathcal{O}\left(N_A^2N_C^2N_{\text{train}}\right)$ and $\mathcal{O}\left(N_A^3N_C^3\right)$, respectively.
Secondly, the bits are allocated to the principal components through Alg.~\ref{alg:bit-allocation}, whose complexity is upper bounded by $\mathcal{O}\left(B\left(b_1+1\right)N_{\text{train}}2^{b_1+1}\right)$.
In fact, Alg.~\ref{alg:bit-allocation} consists of at most $B\left(b_1+1\right)$ $k$-means clustering trainings, each with complexity $\mathcal{O}\left(nkd\right)$, where $n$ is the number of $d$-dimensional points and $k$ is the number of clusters \cite{bac16}.
In Alg.~\ref{alg:bit-allocation}, we have $n=N_{\text{train}}$, $k\leq2^{b_1}$ and $d=2$.
Third, the quantization levels are determined through $N_P$ $k$-means clustering operations on $\mathbf{Z}_{\text{train}}$.
The cost of generating $\mathbf{Z}_{\text{train}}$ is $\mathcal{O}\left(N_AN_CN_PN_{\text{train}}\right)$ while the complexity of the final $k$-means clustering is $\mathcal{O}\left(N_PN_{\text{train}}2^{b_1+1}\right)$.
These operations are all performed offline at the \gls{bs}.
During the online deployment, the complexity is simply given by $\mathcal{O}\left(N_AN_CN_P\right)$ both at the encoder and at the decoder side.

\section{Offloading Overhead-Aware CSI Feedback Design}
\label{sec:overhead}

When using a \gls{dl}-based or \gls{pca}-based \gls{csi} feedback strategy, the relevant model parameters need to be offloaded from the \gls{bs} to the \gls{ue} after the training phase.
In the case of \gls{dl}-based \gls{csi} feedback, these parameters are the weights of the neurons composing the deep architecture involved.
In the case of \gls{pca}-based \gls{csi} feedback, these parameters are the element of the complex matrix $\mathbf{V}_{N_{P}}$ and the complex vector $\boldsymbol{\mu}_{\text{train}}$.
Additionally, in both cases, the actual feedback is a quantized version of the compressed channel.
When $k$-means clustering is used to minimize the \gls{mse}, also the optimal quantization levels need to be offloaded from the \gls{bs} to the \gls{ue}.
This necessary offloading process causes an additional overhead in the transmission.
Thus, it is crucial to design a \gls{csi} feedback strategy being aware that this overhead should be minimum.
In this section, we propose two modifications to our \gls{pca}-based \gls{csi} feedback strategy to reduce the number of offloaded variables.
The first modification reduces the number of offloaded \gls{pca} parameters, while the second reduces the number of offloaded $k$-means clustering levels.

\subsection{Reduction of the Number of Offloaded Model Parameters}

In our \gls{pca}-based \gls{csi} feedback as proposed in Section~\ref{sec:feedback}, the \gls{ue} needs $\mathbf{V}_{N_{P}}$ in order to compress the channel along the first $N_{P}$ principal components.
Thus, the number of real model parameters which are offloaded is
\begin{equation}
N_{O}^{\text{model}}=2N_{A}N_{C}N_{P}+2N_{A}N_{C},
\label{eqn:no-model}
\end{equation}
where the two additive terms are respectively due to the complex matrix $\mathbf{V}_{N_{P}}$ and the complex vector $\boldsymbol{\mu}_{\text{train}}$.
To decrease this number, we observe that the principal components can be sparsified in the angular-delay domain.
Let us now consider the general case in which the \gls{bs} is a \gls{upa} with dimensions $N_{X}\times N_{Y}$, where $N_{X}N_{Y}=N_{A}$.
Here, we define $\mathbf{v}_{n}\in\mathbb{C}^{N_{A}N_{C}\times 1}$ as the $n$-th principal component, i.e., the $n$-th column of $\mathbf{V}_{N_{P}}$.
This vector $\mathbf{v}_{n}$ can be reshaped into a tensor of dimensions $N_{X}\times N_{Y}\times N_{C}$ in order to highlight its two spatial and frequency domains.
We denote the vectorized \gls{dft3} of such a tensor as $\mathbf{f}_{n}\in\mathbb{C}^{N_{A}N_{C}\times 1}$.
Thus, $\mathbf{f}_{n}$ turns out to be sparse since it contains the information of $\mathbf{v}_{n}$ in the two angular and delay domains.
In this way, each principal component $\mathbf{v}_{n}$ can be approximately reconstructed by considering only the significantly non-zero entries of $\mathbf{f}_{n}$.

\begin{figure*}[t]
    \centering
    \includegraphics[width=0.86\textwidth]{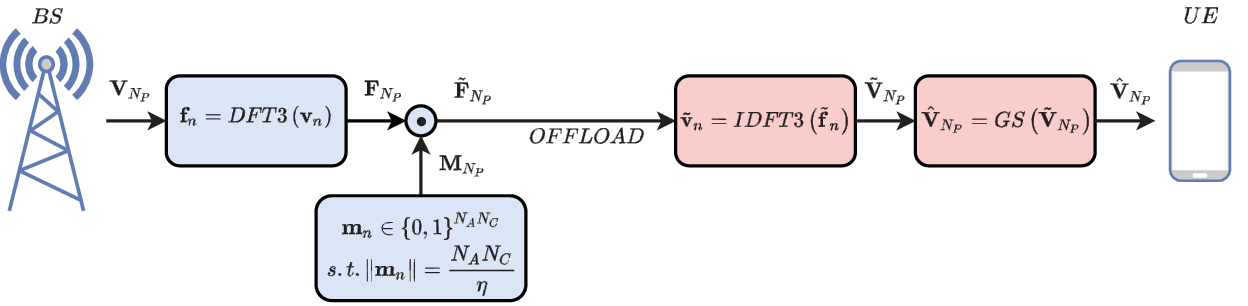}
    \caption{Model parameters offload from the BS to the UE.}
    \label{fig:pca-eta}
\end{figure*}

To decrease the number of offloaded model parameters by exploiting this sparsity property, we introduce a new design parameter $\eta$, such that only a fraction $1/\eta$ of all the $\mathbf{f}_{n}$ entries is offloaded, with $\eta\in\left[1,N_AN_C\right]$.
In detail, let us define the binary mask vector $\mathbf{m}_{n}\in\{0,1\}^{N_{A}N_{C}\times 1}$, containing $\frac{N_{A}N_{C}}{\eta}$ ones in correspondence of the largest values in $\mathbf{f}_{n}$.\footnote{The wording ``largest'' is intended in terms of absolute value, since the $\mathbf{f}_{n}$ elements are complex.}
To retain only the most significant $\frac{N_{A}N_{C}}{\eta}$ values from $\mathbf{f}_{n}$, we consider the Hadamard product $\tilde{\mathbf{f}}_{n}=\mathbf{f}_{n}\odot\mathbf{m}_{n}$.
Thus, the sparse vector $\tilde{\mathbf{f}}_{n}$ is offloaded from the \gls{bs} to the \gls{ue} instead of $\mathbf{v}_{n}$.
In matrix form, we have that $\widetilde{\mathbf{F}}_{N_{P}}=\mathbf{F}_{N_{P}}\odot\mathbf{M}_{N_{P}}$ is offloaded instead of $\mathbf{V}_{N_{P}}$, where $\widetilde{\mathbf{F}}_{N_{P}}=\left[\tilde{\mathbf{f}}_{1},\dots,\tilde{\mathbf{f}}_{N_{P}}\right]$, $\mathbf{F}_{N_{P}}=\left[\mathbf{f}_{1},\dots,\mathbf{f}_{N_{P}}\right]$, and $\mathbf{M}_{N_{P}}=\left[\mathbf{m}_{1},\dots,\mathbf{m}_{N_{P}}\right]$.

At the \gls{ue} side, each vector $\tilde{\mathbf{f}}_{n}$ is firstly reshaped into a tensor of dimensions $N_{X}\times N_{Y}\times N_{C}$.
Then, each principal component $\tilde{\mathbf{v}}_{n}$ is obtained by vectorizing the \gls{idft3} of such a tensor.
Note that the resulting matrix $\widetilde{\mathbf{V}}_{N_{P}}=\left[\tilde{\mathbf{v}}_{1},\dots,\tilde{\mathbf{v}}_{N_{P}}\right]$ is nearly semi-unitary, i.e., $\widetilde{\mathbf{V}}_{N_{P}}^{H}\widetilde{\mathbf{V}}_{N_{P}}\approx\mathbf{I}$, but not semi-unitary since the least significant entries in $\mathbf{F}_{N_{P}}$ have been pruned to reduce the offloading overhead.
For this reason, a last step is needed to ensure that the reconstructed principal component matrix $\widehat{\mathbf{V}}_{N_{P}}=\left[\hat{\mathbf{v}}_{1},\dots,\hat{\mathbf{v}}_{N_{P}}\right]$ is semi-unitary, i.e., $\hat{\mathbf{v}}_{n}\perp \hat{\mathbf{v}}_{m}\;\forall n\neq m$ and $\Vert\hat{\mathbf{v}}_{n}\Vert=1\;\forall n$.
This is achieved by applying the Gram-Schmidt process to the columns of $\widetilde{\mathbf{V}}_{N_{P}}$ yielding $\widehat{\mathbf{V}}_{N_{P}}=GS\left(\widetilde{\mathbf{V}}_{N_{P}}\right)$.\footnote{This Gram-Schmidt process has complexity $\mathcal{O}\left(N_{A}N_{C}N_{P}^{2}\right)$.
Since it is needed only during the offline training stage, we assume that it can be performed at the \gls{ue} with a negligible impact on the energy consumption.
Alternatively, this operation can be replaced by computing the orthogonal factor of the polar decomposition of $\widetilde{\mathbf{V}}_{N_{P}}$ as $\widehat{\mathbf{V}}_{N_{P}}=\widetilde{\mathbf{V}}_{N_{P}}\left(\widetilde{\mathbf{V}}_{N_{P}}^{H}\widetilde{\mathbf{V}}_{N_{P}}\right)^{-1/2}$, or its first order Taylor approximation as $\widehat{\mathbf{V}}_{N_{P}}=\widetilde{\mathbf{V}}_{N_{P}}-\frac{1}{2}\widetilde{\mathbf{V}}_{N_{P}}\left(\widetilde{\mathbf{V}}_{N_{P}}^{H}\widetilde{\mathbf{V}}_{N_{P}}-\mathbf{I}\right)$.
Gram-Schmidt has been selected since it has the lowest complexity and the best reconstruction accuracy among the three.}
To summarize, the proposed modification to reduce the number of offloaded model parameters is graphically represented in Fig.~\ref{fig:pca-eta}.
The channel embedding $\mathbf{z}_{\text{DL}}$ is now computed at the \gls{ue} using  $\widehat{\mathbf{V}}_{N_{P}}$ instead of the exact matrix $\mathbf{V}_{N_{P}}$.
Thus, the optimal bit allocation and the quantization levels should also be computed at the \gls{bs} using the training set projection onto $\widehat{\mathbf{V}}=\left[\hat{\mathbf{v}}_{1},\dots,\hat{\mathbf{v}}_{N_{A}N_{C}}\right]$.
To this end, all the aforementioned operations are applied at the \gls{bs} as part of the offline training process.
More precisely, after the matrix $\mathbf{V}$ is computed with \gls{pca}, it is directly transformed into $\widehat{\mathbf{V}}$ by applying the aforementioned operations.
The only difference with the process depicted in Fig.~\ref{fig:pca-eta} is that in this case the operations are not limited to the first $N_{P}$ principal components.
Subsequently, $\widehat{\mathbf{G}}_{\text{train}}=\mathbf{H}_{\text{train}}\widehat{\mathbf{V}}$ is used in Alg.~\ref{alg:bit-allocation} instead of $\mathbf{G}_{\text{train}}$ to compute the optimal bit allocation and $N_{P}$.
In addition, $\widehat{\mathbf{Z}}_{\text{train}}=\mathbf{H}_{\text{train}}\widehat{\mathbf{V}}_{N_{P}}$ is used to find the optimal quantization levels with $k$-means clustering instead of $\mathbf{Z}_{\text{train}}$.

With this modification, the number of offloaded real model parameters is reduced to
\begin{equation}
N_{O}^{\text{model}}=2\frac{N_{A}N_{C}}{\eta}N_{P}+\frac{N_{A}N_{C}}{\eta}N_{P}+2N_{A}N_{C},
\label{eqn:no-model-mod}
\end{equation}
where the three additive terms are respectively due to the non-zero complex elements of $\widetilde{\mathbf{F}}_{N_{P}}$, the positions of these elements within $\widetilde{\mathbf{F}}_{N_{P}}$, and the complex vector $\boldsymbol{\mu}_{\text{train}}$.
Finally, we remark that this modification allows to introduce an adaptive offloading overhead, as a function of $\eta$.
Thus, $\eta$ could be designed by the network operator as an adaptive parameter to trade offloading overhead impact with channel reconstruction quality.
This degree of freedom is not present in the existing literature where \gls{dl} architectures are used for \gls{csi} feedback.

\subsection{Reduction of the Number of Offloaded $k$-means Clustering Parameters}

Let us now assume that the bit allocation $\mathbf{b}=\left[b_1,\dots,b_{N_{A}N_{C}}\right]$ and $N_{P}$ have been designed through Alg.~\ref{alg:bit-allocation} considering $B$ as the maximum feedback length allowed.
In our \gls{pca}-based \gls{csi} feedback as proposed in Section~\ref{sec:feedback}, the quantization levels are obtained for each principal component, generating $N_{P}$ codebooks which all need to be offloaded from the \gls{bs} to the \gls{ue}.
Thus, to allow the \gls{ue} to generate a feedback of any length less or equal than $B$, we need to offload a codebook including $\sum_{r=1}^{b_{n}}2^{r}$ complex quantization levels for the $n$-th principal component.
In addition, for any feedback length less or equal than $B$, the \gls{ue} needs to know the bit allocation, i.e., how the bits are allocated to the principal components.
To provide the \gls{ue} with all the possible bit allocations, the \gls{bs} offloads the vector $\mathbf{m}=\left[m_1,\dots,m_{B}\right]$, where $m_i$ is the index of the principal component receiving the $i$-th bit in Alg.~\ref{alg:bit-allocation}.
The vector $\mathbf{m}$ is sufficient to describe all the bit allocations with length from $1$ to $B$ because of Proposition~\ref{pro:1}.
Eventually, the number of real $k$-means clustering parameters which are offloaded is
\begin{equation}
N_{O}^{k\text{-means}}=\left(2\sum_{n=1}^{N_{P}}\sum_{r=1}^{b_{n}}2^{r}\right)+B,
\label{eqn:no-kmeans}
\end{equation}
where the two additive terms are respectively due to the complex $k$-means clustering levels and the vector $\mathbf{m}$.

Now, our goal is to decrease the number of offloaded $k$-means clustering levels by considering a unique codebook that can be automatically adapted at the \gls{ue} to fit all the $N_{P}$ principal components.
To this end, we exploit the fact that when a dataset is scaled, also its optimal $k$-means clustering quantization levels experience the same scaling, as formalized in the following proposition.
\begin{proposition}
Let us consider a generic dataset $\mathbf{X}\in\mathbb{R}^{N\times D}$ containing $N$ $D$-dimensional points, and a real scalar $c>0$.
If the optimal $k$-means clustering quantization levels of $\mathbf{X}$ are $\{\mathbf{q}_i\}^*$, the optimal $k$-means clustering quantization levels of the dataset $\mathbf{Y}=c\mathbf{X}$ are $\{\mathbf{q}_{i}^{\prime}\}^*=c\{\mathbf{q}_i\}^*$.
\label{pro:3}
\end{proposition}
\begin{proof}
Please refer to Appendix~C.
\end{proof}

Such a unique codebook is constructed by introducing the matrix $\check{\mathbf{Z}}_{\text{train}}$, obtained by column-wise dividing $\widehat{\mathbf{Z}}_{\text{train}}$ by the vector $\boldsymbol{\sigma}=\left[\sigma_{1},\dots,\sigma_{N_{P}}\right]$.
In other words, all the columns in $\widehat{\mathbf{Z}}_{\text{train}}$ have been normalized in $\check{\mathbf{Z}}_{\text{train}}$ such that they have unitary variance.
Thus, we can now design a codebook which is optimal for all the columns in $\check{\mathbf{Z}}_{\text{train}}$, assuming that their elements are identically distributed.
To do so, we vectorize the matrix $\check{\mathbf{Z}}_{\text{train}}$ into the vector $\check{\mathbf{z}}_{\text{train}}$, and we apply $k$-means clustering to $\check{\mathbf{z}}_{\text{train}}$ considering $k=2^{r}$ with $r\in\left[1,b_{1}\right]$, where $b_{1}$ is the first element of $\mathbf{b}$.
We denote the quantization levels of the resulting codebook as $\{\mathbf{q}_i\}^*$.
Finally, according to Proposition~\ref{pro:3}, we obtain the quantization levels for the $n$-th column of $\widehat{\mathbf{Z}}_{\text{train}}$ (i.e., for the $n$-th principal component) by simply computing $\sigma_{n}\{\mathbf{q}_i\}^*$.

With this modification, the number of real $k$-means clustering parameters which are offloaded is given by
\begin{equation}
N_{O}^{k\text{-means}}=\left(2\sum_{r=1}^{b_{1}}2^{r}\right)+N_{P}+B,
\label{eqn:no-kmeans-mod}
\end{equation}
where the three additive terms are respectively due to the codebook containing the complex quantization levels, the vector $\boldsymbol{\sigma}$ used to adapt the codebook to each principal component, and the vector $\mathbf{m}$ used to obtain the bit allocation.

We finally remark that the number of offloaded parameters considered in this section is to be intended per user.
Nevertheless, in multi-user scenarios, this information is common to all the users connected to the \gls{bs}.
Thus, the model parameters and the $k$-means clustering levels can be broadcasted to all the intended users in these scenarios, without further increasing the offloading overhead.

\section{Numerical Results}
\label{sec:results}

\begin{figure}[t]
    \begin{centering} 
    \includegraphics[width=0.23\textwidth]{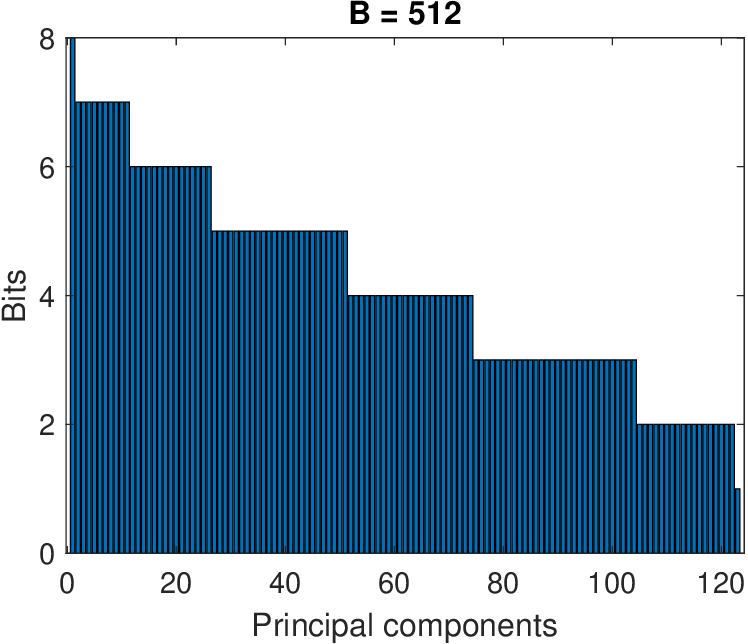}
    \includegraphics[width=0.23\textwidth]{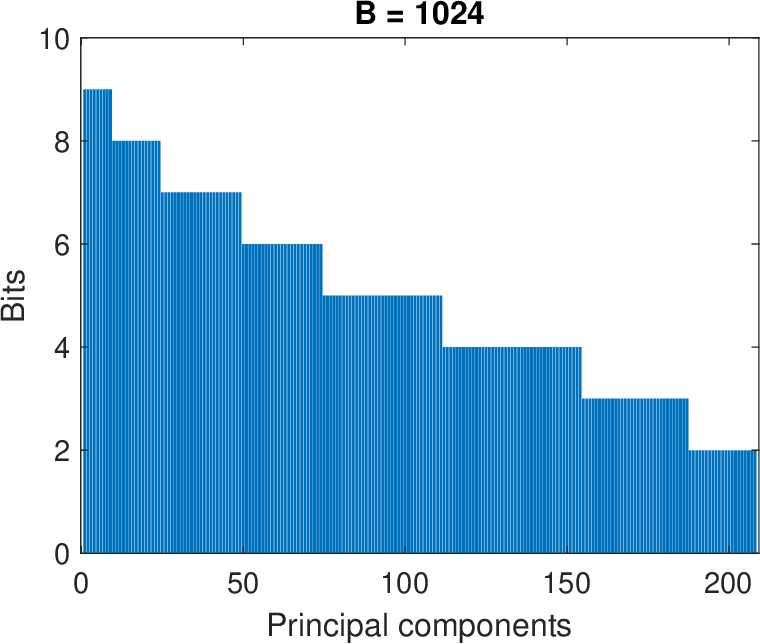}
    \includegraphics[width=0.23\textwidth]{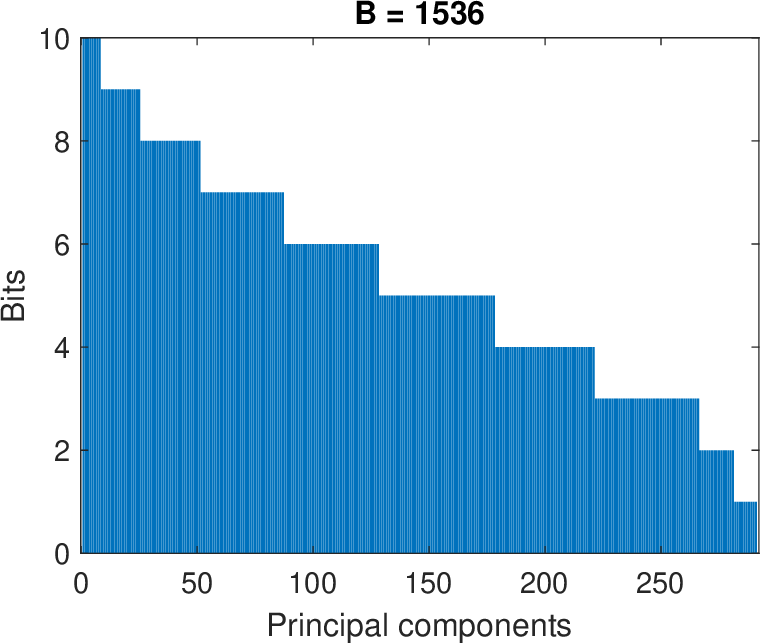}
    \includegraphics[width=0.23\textwidth]{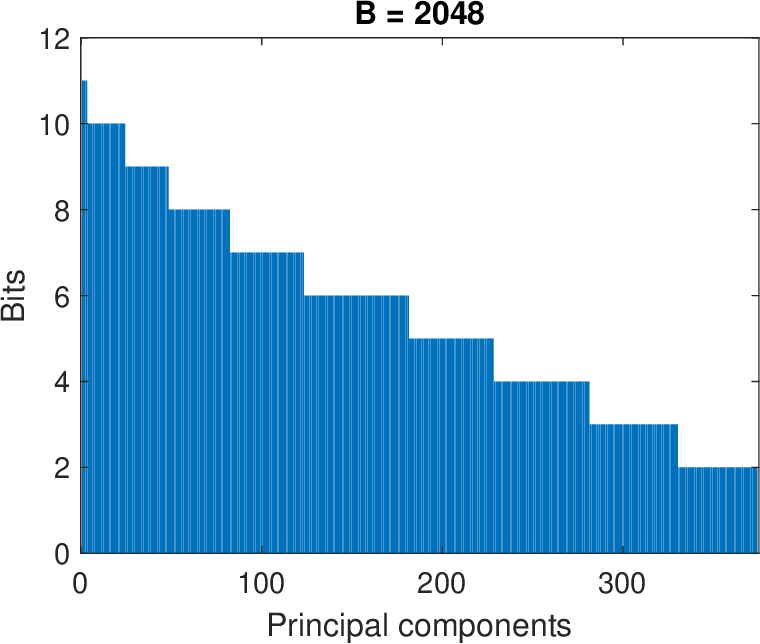}
    \par\end{centering}
    \caption{Bit allocation to the principal components for four different values of $B$.}
    \label{fig:b}
\end{figure}

In this section, we provide numerical results to evaluate our proposed \gls{csi} feedback strategy, and compare it with state-of-the-art \gls{dl} architectures proposed for the same scope.
The performance is measured in terms of downlink channel reconstruction quality, assessed by considering four metrics:
\begin{itemize}
\item the \gls{nmse} between the estimated downlink channel $\widehat{\mathbf{H}}_{\text{DL}}$ and the true downlink channel $\mathbf{H}_{\text{DL}}$, defined as
\begin{equation}
NMSE = \frac{\Vert\widehat{\mathbf{H}}_{\text{DL}}-\mathbf{H}_{\text{DL}}\Vert_F^2}{\Vert\mathbf{H}_{\text{DL}}\Vert_F^2};
\end{equation}
\item the cosine similarity between $\widehat{\mathbf{H}}_{\text{DL}}$ and $\mathbf{H}_{\text{DL}}$, defined as
\begin{equation}
\rho = \frac{1}{N_{C}}\sum_{n_{C}=1}^{N_{C}}\frac{\vert\hat{\mathbf{h}}_{n_{C}}^{H}\mathbf{h}_{n_{C}}\vert}{\Vert\hat{\mathbf{h}}_{n_{C}}\Vert\Vert\mathbf{h}_{n_{C}}\Vert},
\end{equation}
where $\hat{\mathbf{h}}_{n_{C}}$ and $\mathbf{h}_{n_{C}}$ are the $n_{C}$-th columns of $\widehat{\mathbf{H}}_{\text{DL}}$ and $\mathbf{H}_{\text{DL}}$, respectively;
\item the \gls{ber} obtained by precoding \gls{bpsk} symbols at the \gls{bs} based on the reconstructed \gls{csi};
\item the average sum rate obtained with zero-forcing beamforming and water-filling power allocation.
\end{itemize}
In addition, we assess the considered \gls{csi} feedback strategies in terms of offloading overhead (i.e., the number of offloaded model parameters) and in terms of number of training samples required.

\subsection{Dataset Description}

\begin{figure*}[t]
    \begin{centering} 
    \includegraphics[width=0.3\textwidth]{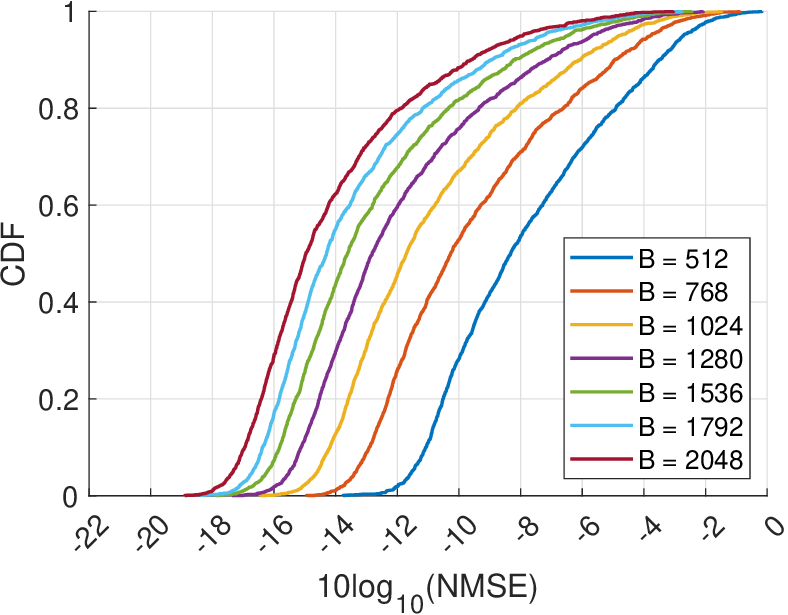}
    \includegraphics[width=0.3\textwidth]{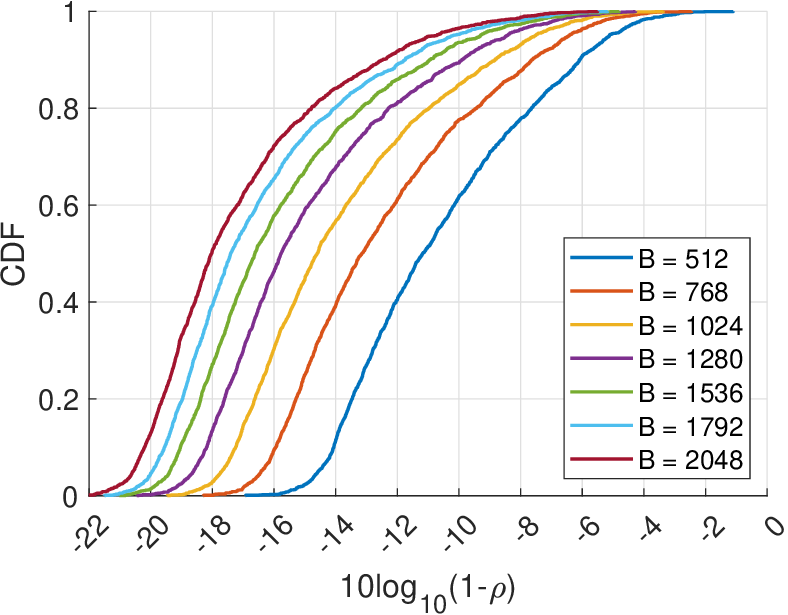}
    \includegraphics[width=0.3\textwidth]{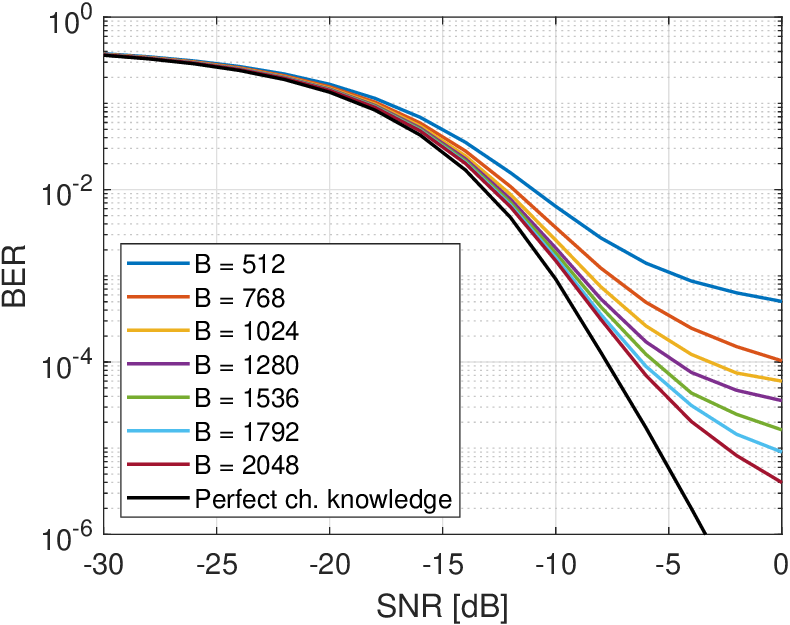}
    \par\end{centering}
    \caption{NMSE CDF (left), cosine similarity CDF (center), and BER (right) for different values of $B$.}
    \label{fig:nmse-rho-ber}
\end{figure*}

The channels for both training and test sets were generated with QuaDRiGa version 2.4, a \textsc{Matlab} based statistical ray-tracing channel simulator \cite{jae14}.
To this end, we consider an urban microcell in which single-antenna \glspl{ue} are uniformly distributed at a distance included in the interval $[20,100]$ m from the \gls{bs}, and at a height of 1.5 m over the ground.
The links between the \gls{bs} and the \glspl{ue} are \gls{nlos}, and $L=58$ paths are considered to simulate a rich multipath environment.
The channel model utilized to generate the channel samples is ``3GPP\_38.901\_UMi\_NLOS'' in QuaDRiGa.
The \gls{bs}, 20 m high, is a \gls{upa} composed of $N_{A}=64$ antennas.
These array elements are arranged on an $N_{X}\times N_{Y}$ shape with an antenna spacing of half wavelength, where $N_{X}=8$ and $N_{Y}=8$.
The uplink and downlink center frequencies are $f_{\text{UL}}=2.5$ GHz and $f_{\text{DL}}=2.62$ GHz.
In each frequency band, the channel bandwidth is $W=8$ MHz, divided into $N_{C}=160$ \gls{ofdm} subcarriers.
The large-scale fading parameters and the path directions are kept constant over the uplink-downlink frequency gap.
On the other hand, the small-scale fading effects depend on the path phase-shifts, in turn dependent on the frequency.
Thus, they are in practice uncorrelated over the frequency gap in rich multipath environments.

The channels resulting from the simulations contain the path loss, the large-scale fading, and small-scale fading effects.
Thus, in order to work with scaled channels, we normalize the dataset according to 
\begin{equation}
\mathbf{H}_{\text{UL}}\leftarrow L_{\text{UL}}\mathbf{H}_{\text{UL}},\:\mathbf{H}_{\text{DL}}\leftarrow L_{\text{DL}}\mathbf{H}_{\text{DL}},
\end{equation}
where the scalars $L_{\text{UL}}^{-1}$ and $L_{\text{DL}}^{-1}$ contain the path loss and the large-scale fading effects.
In total, $N=N_{\text{train}}+N_{\text{test}}$ users are randomly dropped in the cell and, for each user, the pair of channels $\mathbf{H}_{\text{UL}}$, $\mathbf{H}_{\text{DL}}$ is generated.
Firstly, the \emph{offline learning} stage is carried out considering $N_{\text{train}}$ noisy uplink channel matrices.
Secondly, the performance is tested on $N_{\text{test}}$ downlink channel matrices, corresponding to the users not considered during the training phase.
We assume that only noisy versions of the channel matrices are available both in the \emph{offline learning} and \emph{online development} stages, with $SNR_{\text{UL}}=10$ dB and $SNR_{\text{DL}}=10$ dB.
The ground truth channels $\mathbf{H}_{\text{DL}}$ are only used to evaluate the reconstruction performance.
Furthermore, we set $N_{\text{train}}=5000$ and $N_{\text{test}}=2000$.

\subsection{Performance Evaluation}

We now assess the performance of our \gls{pca}-based \gls{csi} feedback strategy.
Experimentally, we noticed that the two modifications proposed in Section~\ref{sec:overhead} do not impact visibly the performance if $\eta\approx 16$ or less.
Thus, the performance reported here refers to our proposed \gls{csi} feedback strategy including the two modifications with $\eta=16$, unless otherwise specified.
First of all, we inspect how the feedback bits are allocated to the principal components.
In Fig.~\ref{fig:b}, we report a graphical representation of the vector $\mathbf{b}$ resulting from Alg.~\ref{alg:bit-allocation} applied to four different values of $B$.
As expected, we observe that the number of considered principal components $N_{P}$ and the number of bits assigned to the first principal component $b_{1}$ increase as $B$ increases.

Fig.~\ref{fig:nmse-rho-ber} reports the \glspl{cdf} of the \gls{nmse} and the cosine similarity $\rho$ between the reconstructed and the true downlink channel matrices.
Besides, we plot the \gls{ber} obtained in downlink by precoding uncoded \gls{bpsk} symbols according to the reconstructed \gls{csi}.
The \glspl{cdf}, calculated over the whole test set, and the obtained \gls{ber} show that the reconstruction quality increases with $B$, and that, more importantly, no performance upper bound is present.
When the \gls{ber} metric is considered, the only performance upper bound present is given by the \gls{ber} obtained with perfect \gls{csi}.
In fact, $B$ could be potentially increased until all the principal components are considered to embed the \gls{csi}.
In this extreme case, approximately perfect channel reconstruction could be achieved.
This property cannot be found in \gls{dl} strategies recently proposed for \gls{csi} feedback.
In these strategies, the reconstruction quality is always bounded by the latent space dimensionality, which is fixed regardless of the value of $B$.

Now, we investigate how the channel reconstruction quality varies according to the antenna number $N_A$ and subcarrier number $N_C$.
To this end, we consider four different \gls{csi} dimensions $N_A\times N_C$, namely $32\times 80$, $32\times 160$, $64\times 80$, and $64\times 160$.
We show the reconstruction quality in terms of cosine similarity obtained for these channels when $B=512$ and $B=1024$ feedback bits are used in Fig.~\ref{fig:NANC-Ntrain} (left).
As expected, the reconstruction quality is higher for \gls{csi} matrices with lower dimensionalities.
We also observe that the channels with dimensions $32\times 160$ are better reconstructed than the channels with dimensions $64\times 80$, despite these two types of channel matrices have the same number of entries.
The reason is that frequency correlation is higher than spatial correlation in an \gls{ofdm} channel matrix, as also noticed in \cite{wen18}.
Thus, the information along the frequency dimension can be compressed more efficiently.

A problem concerning \gls{ml} approaches is the need of training samples.
Experimentally, we observe that $N_{\text{train}}=5000$ training samples are necessary to well-train both \gls{pca} and $k$-means clustering.
For this reason, this is the number of training samples used to present the obtained results in the manuscript.
However, we explore now the performance degradation experienced when this number is lowered.
In Fig.~\ref{fig:NANC-Ntrain} (right), the \gls{cdf} of the cosine similarity obtained with different values of $N_{\text{train}}$ is reported.
We observe that the reconstruction quality only slightly degrades when $B=1536$, while it remains approximately stable when the number of bits employed is low, i.e., $B=512$.
This means that the principal components are well reconstructed even with less training samples.
However, $k$-means clustering is the bottleneck that precludes the use of a lower $N_{\text{train}}$, especially when the number of quantization levels is high.

\begin{figure}[t]
    \begin{centering} 
    \includegraphics[width=0.23\textwidth]{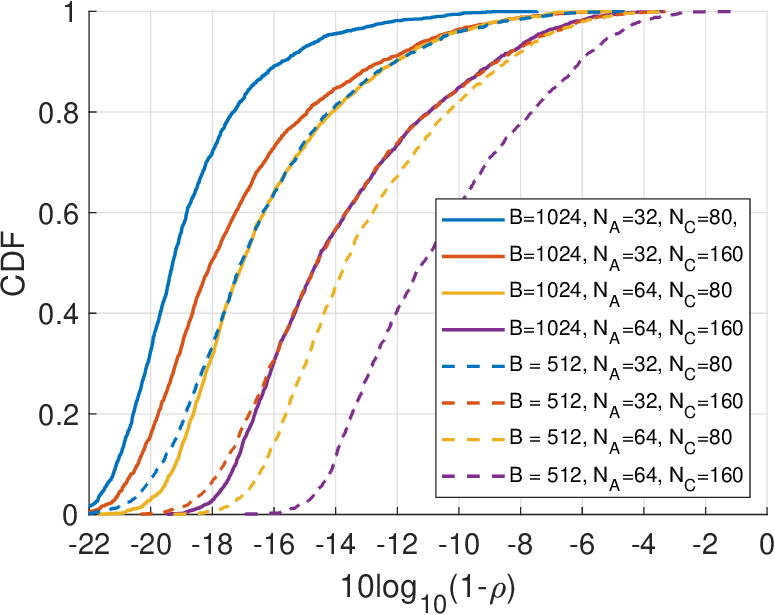}
    \includegraphics[width=0.23\textwidth]{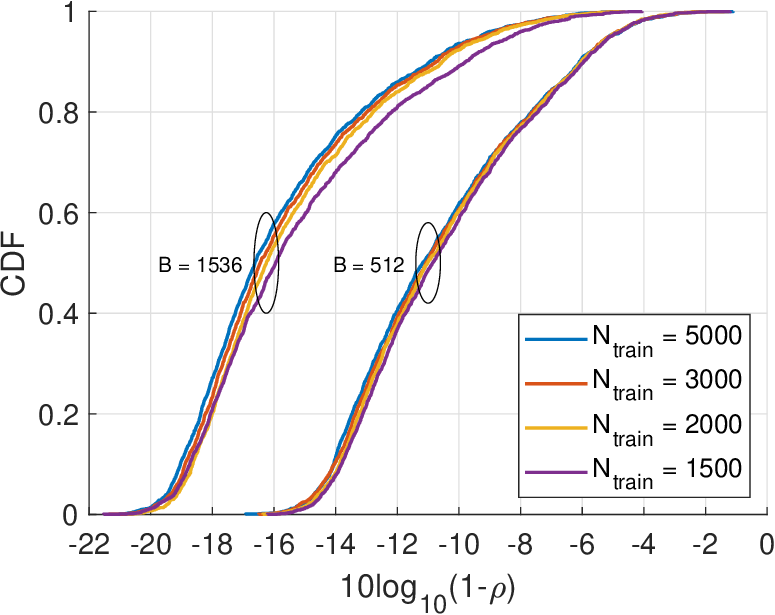}
    \par\end{centering}
    \caption{Cosine similarity CDFs for different CSI dimensions (left) and for different number of training samples (right).}
    \label{fig:NANC-Ntrain}
\end{figure}

Finally, we report the performance of multi-user precoding carried out with the reconstructed downlink channel matrices.
More precisely, we use the reconstructed downlink channels to serve $K=8$ \glspl{ue} through zero-forcing beamforming with water-filling power allocation.
Zero-forcing beamforming is applied independently to each subcarrier and the resulting sum rate is averaged over the $N_{C}$ subcarriers.
The water-filling power allocation is designed assuming that the reconstructed downlink channel is the perfect channel.
Thus, the multi-user interference is not captured in the water-filling solution.
In order to obtain reliable results, Monte Carlo simulations are run by randomly selecting, in each simulation, $K=8$ downlink channels among the $N_{\text{test}}$ available in the test set.
Fig.~\ref{fig:zfbf} reports the average sum rate obtained by compressing the \gls{csi} with four different feedback lengths $B$.
Here, our \gls{csi} feedback strategy, namely ``PCA'', is compared with two baseline strategies based on \gls{dl} architectures.
``AE'' is the convolutional autoencoder proposed in \cite{riz21}, while ``CsiNetPro'' is the architecture proposed in \cite{li20}.
In addition, the learning-based results are compared with the classical theory-based approach ``IDFT'', also used for comparison in \cite{riz21}.
According to the ``IDFT'' approach, the noisy \gls{csi} $\widetilde{\mathbf{H}}_{\text{DL}}$ is firstly transformed into the space-delay domain through a \gls{idft2}.
Then, only the elements of the first columns of the resulting matrix are considered for the feedback.
At the BS, the \gls{csi} is reconstructed through a zero-padding operation followed by a \gls{dft2}.
Note that depending on how many columns are retained, the \gls{csi} can be compressed into latent spaces with different dimensionalities.
For these three baselines, the number in the legend name denotes the latent space dimensionality.
For conciseness, in Fig.~\ref{fig:zfbf} we report the performance of ``IDFT'' applied with the best latent space dimensionality for each $B$.

\begin{figure}[t]
    \begin{centering} 
    \includegraphics[width=0.23\textwidth]{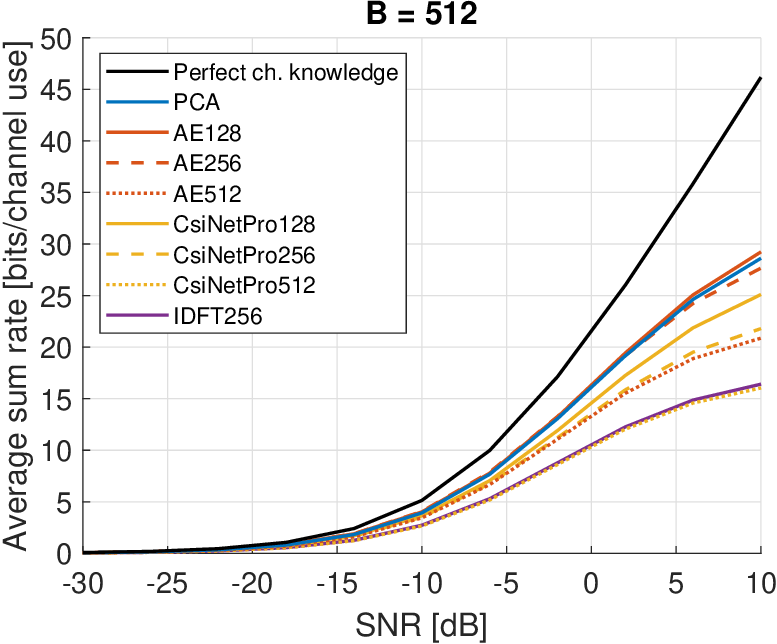}
    \includegraphics[width=0.23\textwidth]{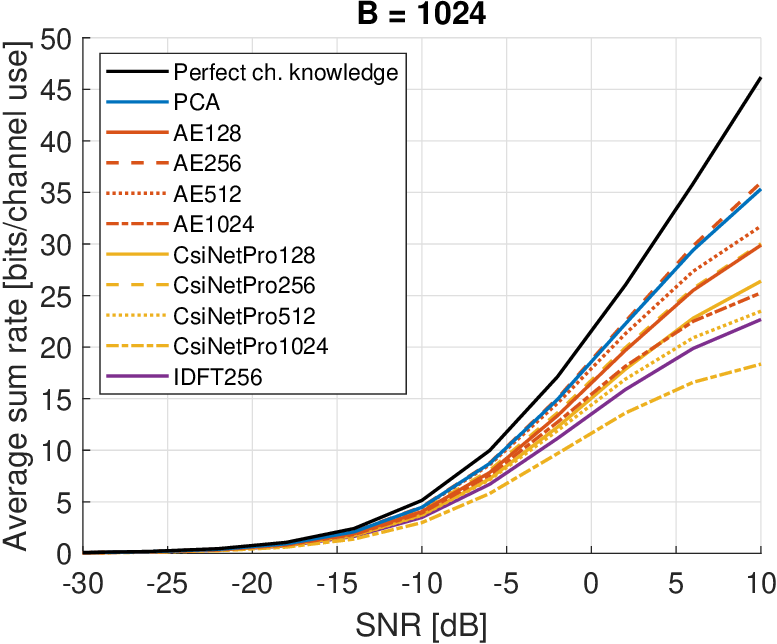}
    \includegraphics[width=0.23\textwidth]{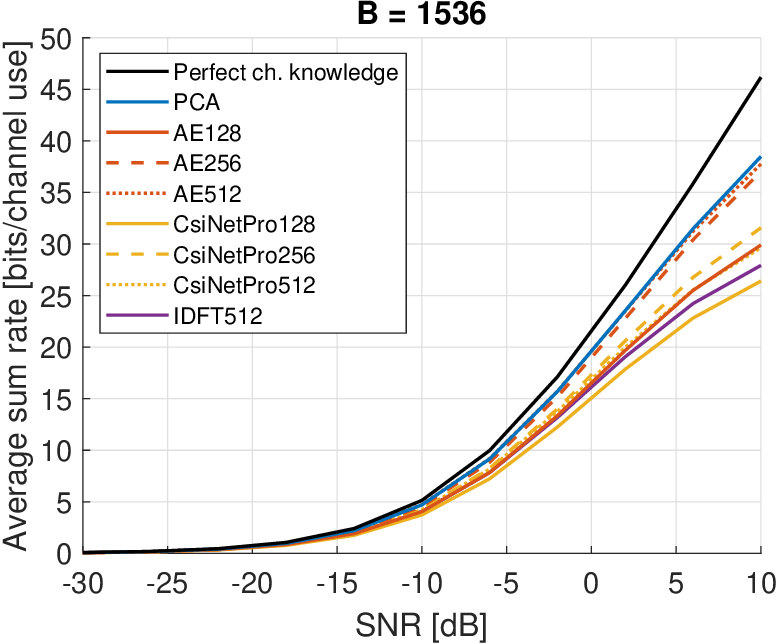}
    \includegraphics[width=0.23\textwidth]{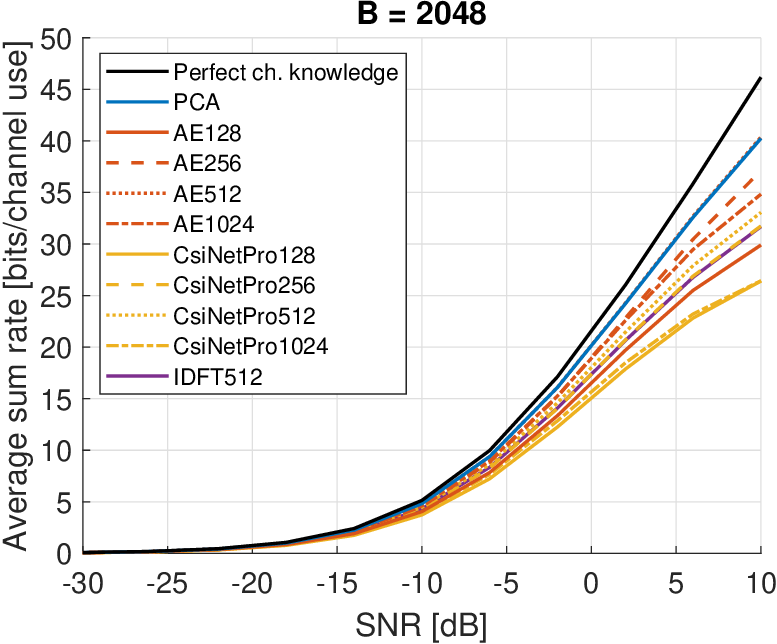}
    \par\end{centering}
    \caption{Average sum rate with zero-forcing beamforming for four different values of $B$. ``PCA'' has been trained with 2000 samples, ``AE'' and ``CsiNetPro'' have been trained with 40000 samples.}
    \label{fig:zfbf}
\end{figure}

For the baseline strategies, the feedback bits are allocated uniformly to all the latent space dimensions, and the quantization levels are determined with $k$-means clustering\footnote{The proposed bit allocation strategy could be applied also to \gls{dl} architectures.
However, in these architectures, the compressed \gls{csi} entries are approximately identically distributed for all latent space dimensions.
For this reason, the proposed bit allocation applied to autoencoders boils down to uniform bit allocation, with no visible improvement with respect to the latter.}.
Preliminary experiments confirmed that $N_{\text{train}}=5000$ training samples are not enough to well-train the ``AE'' and ``CsiNetPro'' architectures.
With this amount of training data, these \gls{dl} architectures are not able to learn an expressive representation of the channels in the latent space.
For this reason, in Fig.~\ref{fig:zfbf}, we report the performance of well-trained ``AE'' and ``CsiNetPro'' architectures, i.e., trained with 40000 training samples.

From Fig.~\ref{fig:zfbf}, we notice that the learning-based solutions always outperform ``IDFT''.
Among the two considered \gls{dl} architectures, ``AE'' outperforms ``CsiNetPro'' for every considered $B$, as already highlighted in \cite{riz21}.
Furthermore, for the \gls{dl} architectures, we notice that there is not a unique latent space dimensionality $N_{L}$ that is optimal for every feedback length $B$.
This result supports our intuition that $N_{L}$ should be designed according to the number of bits $B$.
Since we consider such an adaptive design, our strategy ``PCA'' performs approximately as the best autoencoder ``AE'' in every $B$ considered.
Thus, \gls{pca} can be successfully applied to compress the channel matrix with a significantly reduced training set.

\begin{figure}[t]
    \centering
    \includegraphics[width=0.44\textwidth]{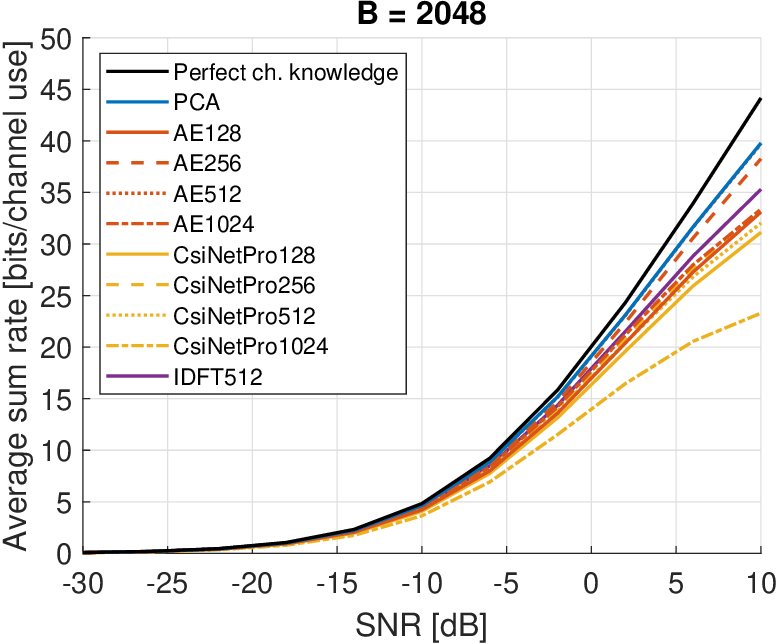}
    \caption{Average sum rate with zero-forcing beamforming for $B=2048$. The \gls{ml} models have been trained on urban microcell channel samples and tested on rural macrocell channels.}
    \label{fig:zfbf-UMaNLOS}
\end{figure}

In this study and in the related literature, it is assumed that the trained \gls{ml} algorithm is used to compress channel matrices with the same fading distribution as seen during the training session.
However, when the channel statistics vary, the \gls{csi} feedback strategy should be able to compress differently distributed channels before an updated trained model becomes available.
To explore the performance under different channel models, we train our PCA-based approach and the baseline \gls{dl} architectures on uplink channel samples drawn from the urban microcell scenario ``3GPP\_38.901\_UMi\_NLOS''.
After this training session, we test the \gls{csi} feedback methods on downlink channel samples drawn from the rural macrocell scenario ``3GPP\_38.901\_RMa\_NLOS''.
Note that the multipath in the latter environment is less severe than in the former, since only $L=11$ paths are considered by the QuaDRiGa simulator.
The only elements in common between the training and testing scenarios are the number of antennas and subcarriers, the bandwidth and the center frequencies of uplink and downlink bands.
The obtained sum rate is reported in Fig.~\ref{fig:zfbf-UMaNLOS}, for $B=2048$.
Fig.~\ref{fig:zfbf-UMaNLOS} shows that the knowledge obtained in a rich multipath environment can be successfully reused to compress sparser \gls{csi} matrices.
Our \gls{pca}-based strategy performs as the best ``AE'' architecture, i.e., ``AE512'', and largely outperforms ``CsiNetPro'' and ``IDFT512'' baselines.

To analyze the effect of the parameter $\eta$ on the performance, we plot the average sum rate versus the feedback length $B$, by fixing the \gls{snr} experienced by each of the $K$ user to $10$ dB.
In Fig.~\ref{fig:zfbf-snr10}, the performance of our strategy is compared with the two considered baseline \gls{dl} architectures, each with four different $N_{L}$.
We notice that the performance of each \gls{dl} architecture saturates as $B$ increases due to the fixed latent space dimensionality.
Thus, each $N_{L}$ turns out to be optimal only for some specific values of $B$.
Mode precisely, ``AE128'' is the best autoencoder when $B=512$, ``AE256'' is preferred when $B\in\{768,1024,1280\}$, while ``AE512'' is the autoencoder achieving the highest average sum rate when $B\in\{1536,1792,2048\}$.
The four ``CsiNetPro'' achieve a lower average sum rate than their respective ``AE'' architectures.
Conversely, the performance of our \gls{pca}-based strategy with $\eta=16$ is comparable with the one of the best autoencoder ``AE'' for each value of $B$.
We notice that the performance deterioration caused by increasing $\eta$ is minimal.

\begin{figure}[t]
    \centering
    \includegraphics[width=0.44\textwidth]{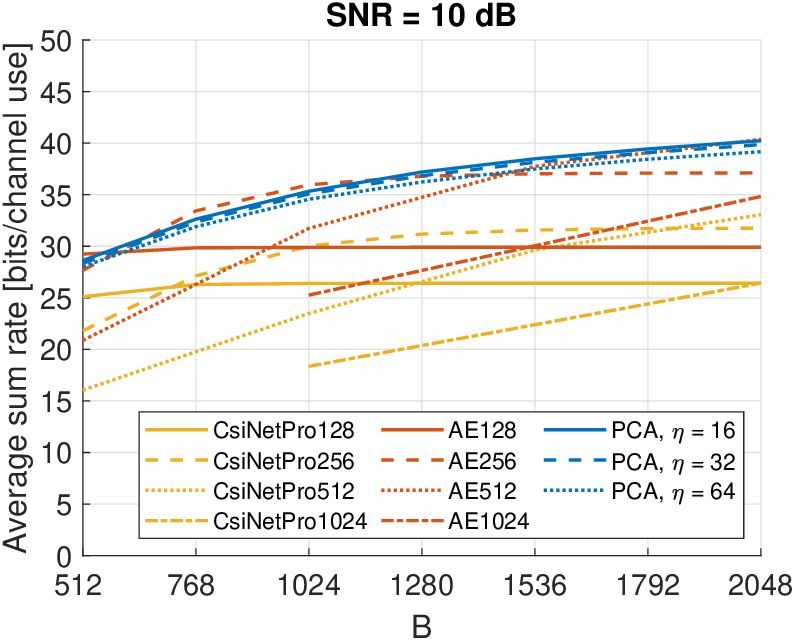}
    \caption{Average sum rate with zero-forcing beamforming vs feedback length $B$, with \gls{snr} $=10$ dB at each \gls{ue}. ``PCA'' has been trained with 2000 samples, ``AE'' and ``CsiNetPro'' have been trained with 40000 samples.}
    \label{fig:zfbf-snr10}
\end{figure}

\subsection{Offloading Impact}

We now analyze the impact of the two modifications proposed in Section~\ref{sec:overhead} on the overhead caused by parameter offloading.
The effect of the first modification can be analyzed by comparing \eqref{eqn:no-model} with \eqref{eqn:no-model-mod}.
When computing $N_{O}^{\text{model}}$, the number of offloaded parameters due to $\boldsymbol{\mu}_{\text{train}}$ can be neglected, since in practice $N_{P}\gg 1$.
Thus, we have that the first modification decreases the number of offloaded model parameters of approximately $\frac{2}{3}\eta$ times.
In our numerical scenario, $B=2048$ yields $N_{P}=374$.
Hence, without the first modification, we have $N_{O}^{\text{model}}=7.68\times 10^6$ according to \eqref{eqn:no-model}.
Conversely, when the first modification is applied with $\eta=16$, $N_{O}^{\text{model}}$ is reduced to $0.739\times 10^6$ according to \eqref{eqn:no-model-mod}.

Considering the offloaded $k$-means clustering parameters, the effect of the second modification is given by comparing \eqref{eqn:no-kmeans} with \eqref{eqn:no-kmeans-mod}.
Also to quantify the effect of the second modification, we refer to the case in which the maximum feedback length allowed is $B=2048$, yielding $N_{P}=374$ and $b_{1}=11$ in our numerical scenario.
Without the second modification, we have $N_{O}^{k\text{-means}}=243\times 10^3$ according to \eqref{eqn:no-kmeans}.
Conversely, when the second modification is applied, $N_{O}^{k\text{-means}}$ is reduced to $10.6\times 10^3$ according to \eqref{eqn:no-kmeans-mod}.

\begin{figure}[t]
    \centering
    \includegraphics[width=0.44\textwidth]{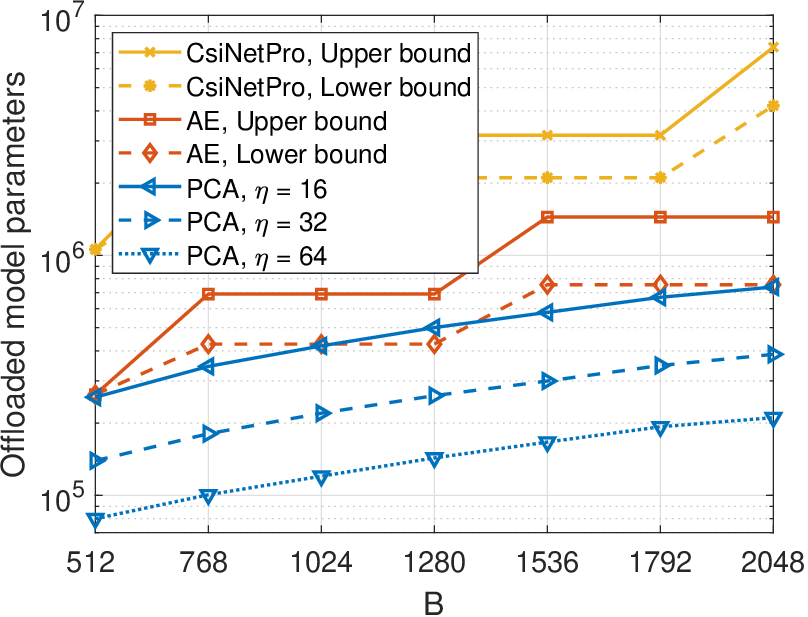}
    \caption{Number of offloaded model parameters vs feedback length $B$. ``PCA'' has been trained with 2000 samples, ``AE'' and ``CsiNetPro'' have been trained with 40000 samples.}
    \label{fig:offload}
\end{figure}

When both modifications are applied, we notice that $N_{O}^{\text{model}}\gg N_{O}^{k\text{-means}}$.
Thus, now we compare our \gls{pca}-based \gls{csi} feedback strategy with the two baseline \gls{dl}-based strategies in terms of offloaded model parameters, neglecting the offloaded $k$-means clustering parameters.
This comparison is carried out in Fig.~\ref{fig:offload}, where $N_{O}^{\text{model}}$ is reported for our strategy when both modifications are considered, with $\eta\in\{16,32,64\}$.
For the two considered \gls{dl} architectures, an upper bound and a lower bound on the number of offloaded model parameters are reported.
The upper bound represents the case in which we offload all the necessary encoders to generate a feedback of length less or equal than $B$.
These encoders can be retrieved from Fig.~\ref{fig:zfbf-snr10}, considering only the most performing ones for feedback lengths less or equal than $B$.
The lower bound represents the case in which we apply the parallel multiple-rate framework proposed in \cite{guo20a}.
In this framework, among the necessary encoders, only the encoder with the largest latent space dimensionality is offloaded.
The other encoders are obtained at the \gls{ue} by just considering a reduced number of latent space dimensions from this encoder.
Note that this framework inevitably causes a slight degradation in the channel reconstruction performance, as stated in \cite{guo20a}, that is not assessed in this study.
We remark that the parameter number of \gls{dl} models can be reduced at the cost of a performance degradation with techniques such as pruning and knowledge distillation.
However, to offer a clear comparison in terms of both reconstruction performance and offloading overhead, we do not alter the architectures proposed in \cite{riz21,li20}.

\begin{table*}[t]
\centering
\caption{Performance comparison with feedback length $B=2048$.}
\begin{tabular}{@{}cccc@{}}
\toprule
               & \begin{tabular}[c]{@{}c@{}}Average sum rate\\$[$bits/channel use$]$\end{tabular} &
                 \begin{tabular}[c]{@{}c@{}}Number of\\offloaded model parameters\end{tabular} &
                 \begin{tabular}[c]{@{}c@{}}Number of\\training samples\end{tabular}\\
\midrule
PCA, $\eta=16$ & 40.2 & $0.739\times 10^6$  & $5\times 10^3$\\
PCA, $\eta=32$ & 39.9 & $0.387\times 10^6$  & $5\times 10^3$\\
PCA, $\eta=64$ & 39.2 & $0.211\times 10^6$  & $5\times 10^3$\\
\midrule
AE             & 40.4 & $1.45\times 10^6$ & $40\times 10^3$\\
CsiNetPro & 33.1 & $7.37\times 10^6$ & $40\times 10^3$\\
\bottomrule
\end{tabular}
\label{tab:summary}
\end{table*}

Fig.~\ref{fig:offload} shows that our strategy is convenient also in terms of offloading overhead compared with the two \gls{dl}-based strategies, not only in terms of training samples required.
When $\eta=16$, our strategy and the ``AE'' strategy used with the framework proposed in \cite{guo20a} cause similar offloading overhead.
However, it is possible to significantly reduce this overhead of our strategy by only slightly affecting the reconstruction performance thanks to the adaptive parameter $\eta$.

Finally, in Tab~\ref{tab:summary} we summarize the comparison between our \gls{pca}-based strategy and the two considered \gls{dl}-based strategies, for a feedback length $B=2048$.
Here, the offloaded parameter number for the deep architectures are computed without considering the framework in \cite{guo20a}, since it would cause a non-negligible performance degradation.
Compared to ``AE'', our strategy using $\eta=16$ halves the offloading overhead and reduces the number of training samples by eight times.
If $\eta=64$ is considered, the offloading overhead is reduced by 6.87 times, at the cost of a 2.97\% average sum rate loss.
Compared to ``CsiNetPro'', our strategy improves the average sum rate by 21.7\% (resp. 18.5\%), and reduces the offloading overhead by 10.0 (resp. 34.9) times, when $\eta=16$ (resp. $\eta=64$).
Considering maximum feedback lengths ranging from $B=521$ to $B=2048$, our strategy used with $\eta=64$ improves the sum rate by 17\% on average, and reduces the offloading overhead by 23.4 times compared to ``CsiNetPro''.

\section{Conclusion}
\label{sec:conclusion}

In this study, we propose a novel strategy to design the \gls{csi} feedback in \gls{fdd} massive \gls{mimo} systems.
This strategy allows to design the feedback with variable length, while reducing the number of parameter offloaded from the \gls{bs} to the \gls{ue}.
Firstly, the channel is compressed using \gls{pca}, with a latent space dimensionality adapted to the number of available feedback bits.
Then, the feedback bits are allocated to the principal components by minimizing a properly defined \gls{nmse} distortion.
Finally, the quantization levels are determined with $k$-means clustering.
In addition, we allow an adaptive number of offloaded model parameters, which can be adjusted to trade offloading overhead and \gls{csi} reconstruction quality.
Such an adaptive offloading overhead has been never considered in previous literature employing \gls{dl} approaches.

Through simulations, we compare our strategy with state-of-the-art \gls{dl} architectures proposed for the same scope.
Numerical results show that our strategy performs better or approximately equal, to \gls{dl} architectures well-trained on larger datasets, in terms of sum rate obtained with multi-user precoding.
The offloading overhead can be significantly reduced in our strategy, with approximately no impact on the \gls{csi} reconstruction.
At the same time, \gls{pca} is characterized by a lightweight training phase, requiring a reduced number of training samples.
This lightweight training phase enables, in practical developments, more frequent trainings.
In this way, the compression strategy could be better maintained updated, as the environment evolves in time.
Compared to ``CsiNetPro'', our strategy using $\eta=64$ improves the sum rate by 17\% on average, reduces the offloading overhead by 23.4 times, and requires eight times fewer training parameters.

\appendix

\subsection{Proof of Proposition~\ref{pro:1}}
\label{subsec:proof1}

Since the bit allocation $\mathbf{b}^\prime=\left[b_1^\prime,\dots,b_{N_{A}N_{C}}^\prime\right]$ contains one more bit than $\mathbf{b}=\left[b_{1},\dots,b_{N_{A}N_{C}}\right]$, there is always a principal component $m$ such that $b_{m}^{\prime}=b_{m}+C$, with $C\geq1$.
Thus, to prove Proposition~\ref{pro:1}, we need to prove that the bit allocation $\left[b_{1},\dots,b_{m}+1,\dots,b_{N_{A}N_{C}}\right]$ is better than any other allocation of $B+1$ bits $\left[b_1^\prime,\dots,b_m+C,\dots,b_{N_{A}N_{C}}^\prime\right]$.
This is equivalent to saying that the distortion caused by $\left[b_{1},\dots,b_{m}+1,\dots,b_{N_{A}N_{C}}\right]$ is less than the distortion caused by $\left[b_1^\prime,\dots,b_m+C,\dots,b_{N_{A}N_{C}}^\prime\right]$, that is
\begin{multline}
d_{1}\left(b_1\right)+\dots+d_{m}\left(b_m+1\right)+\dots+d_{N_{A}N_{C}}\left(b_{N_{A}N_{C}}\right)\\
<d_{1}\left(b_1^\prime\right)+\dots+d_{m}\left(b_m+C\right)+\dots+d_{N_{A}N_{C}}\left(b_{N_{A}N_{C}}^\prime\right)\label{eqn:proof1-1}
\end{multline}

For convenience, we rewrite \eqref{eqn:proof1-1} as
\begin{multline}
\left[d_{1}\left(b_1\right)+\dots+d_{m}\left(b_m\right)+\dots+d_{N_{A}N_{C}}\left(b_{N_{A}N_{C}}\right)\right]\\
+\left[d_{m}\left(b_m+1\right)-d_{m}\left(b_m\right)\right]\\
<\left[d_{1}\left(b_1^\prime\right)+\dots+d_{m}\left(b_m+C-1\right)+\dots+d_{N_{A}N_{C}}\left(b_{N_{A}N_{C}}^\prime\right)\right]\\
+\left[d_{m}\left(b_m+C\right)-d_{m}\left(b_m+C-1\right)\right],\label{eqn:proof1-2}
\end{multline}
where two additive terms are highlighted in both sides of the inequality.
Thus, \eqref{eqn:proof1-2} can be proved by independently verifying the following two inequalities
\begin{multline}
d_{1}\left(b_1\right)+\dots+d_{m}\left(b_m\right)+\dots+d_{N_{A}N_{C}}\left(b_{N_{A}N_{C}}\right)\\
<d_{1}\left(b_1^\prime\right)+\dots+d_{m}\left(b_m+C-1\right)+\dots+d_{N_{A}N_{C}}\left(b_{N_{A}N_{C}}^\prime\right)\label{eqn:proof1-3}
\end{multline}
\begin{equation}
d_{m}\left(b_m+1\right)-d_{m}\left(b_m\right)\leq d_{m}\left(b_m+C\right)-d_{m}\left(b_m+C-1\right).
\label{eqn:proof1-4}
\end{equation}
Firstly, \eqref{eqn:proof1-3} holds since $\mathbf{b}=\left[b_{1},\dots,b_{N_{A}N_{C}}\right]$ is the optimal bit allocation of $B$ bits.
Thus, any other allocation of $B$ bits $\left[b_1^\prime,\dots,b_m+C-1,\dots,b_{N_{A}N_{C}}^\prime\right]$ causes an higher distortion.
Secondly, to prove \eqref{eqn:proof1-4}, we resort to an explicit expression of the \gls{mse} distortion as a function of the bit number.
Since such a distortion in the case of $k$-means clustering is not available in close form, we consider the distortion-rate function, which provides a lower bound on the \gls{mse} distortion \cite{cov99}.
According to information theory, the distortion-rate function of a \gls{cscg} random variable $X_{n}\sim\mathcal{CN}\left(0,\sigma_{n}^{2}\right)$ quantized with $b_{n}$ bits is given by $d_{n}\left(b_{n}\right)=\sigma_{n}^{2}2^{-b_{n}}$ \cite{cov99}.
Thus, assuming that on the $n$-th principal component the training set is distributed as $\mathcal{CN}\left(0,\sigma_{n}^{2}\right)$, \eqref{eqn:proof1-4} becomes
\begin{align}
\sigma_m^{2}2^{-\left(b_m+1\right)}-\sigma_m^{2}2^{-b_m}&\leq\sigma_m^{2}2^{-\left(b_m+C\right)}-\sigma_m^{2}2^{-\left(b_m+C-1\right)}\\
\sigma_m^{2}2^{-b_m}\left(2^{-1}-1\right)&\leq\sigma_m^{2}2^{-\left(b_m+C-1\right)}\left(2^{-1}-1\right),
\end{align}
which is verified since $C\geq1$ implies $2^{-b_m}\geq 2^{-\left(b_m+C-1\right)}$.

\subsection{Proof of Proposition~\ref{pro:2}}
\label{subsec:proof2}

To prove Proposition~\ref{pro:2}, we prove that the optimal $b_{n}$ cannot be greater than $b_{n-1}$ since the bit allocation $b_{n-1}=C+\Delta C$, $b_{n}=C$ is always better than $b_{n-1}=C$, $b_{n}=C+\Delta C$, where $C\in\mathbb{N}$, $\Delta C\in\mathbb{N^*}$.
This means that
\begin{equation}
d_{n-1}\left(C+\Delta C\right)+d_{n}\left(C\right)<d_{n-1}\left(C\right)+d_{n}\left(C+\Delta C\right),
\label{eqn:proof2-1}
\end{equation}
since the best bit allocation to the two principal components $n-1$ and $n$ is the one that minimizes the distortion $d_{n-1}+d_{n}$.
As in Appendix~A, we write the distortion caused by quantizing the $n$-th principal component with $b_{n}$ bits as $d_{n}\left(b_{n}\right)=\sigma_{n}^{2}2^{-b_{n}}$.
Thus, \eqref{eqn:proof2-1} can be rewritten as
\begin{align}
\sigma_{n-1}^{2}2^{-\left(C+\Delta C\right)}+\sigma_{n}^{2}2^{-C}&<\sigma_{n-1}^{2}2^{-C}+\sigma_{n}^{2}2^{-\left(C+\Delta C\right)}\\
\sigma_{n}^{2}2^{-C}-\sigma_{n}^{2}2^{-\left(C+\Delta C\right)}&<\sigma_{n-1}^{2}2^{-C}-\sigma_{n-1}^{2}2^{-\left(C+\Delta C\right)}\\
\sigma_{n}^{2}2^{-C}\left(1-2^{-\Delta C}\right)&<\sigma_{n-1}^{2}2^{-C}\left(1-2^{-\Delta C}\right),
\end{align}
which is verified since the principal components are ordered, i.e., $\sigma_{n}^{2}<\sigma_{n-1}^{2}$.

\subsection{Proof of Proposition~\ref{pro:3}}
\label{subsec:proof3}

$k$-means clustering groups the $N$ dataset points into $k$ clusters, each assigned to a unique quantization level.
This is done by minimizing the sum of the square distances between each data point and the quantization level assigned to its cluster \cite{bis06}.
Let us denote with $\mathbf{x}_{n}$ the $n$-th point in the dataset $\mathbf{X}$.
The $k$-means clustering process can be formalized by introducing the binary indicator $r_{ni}\in\{0,1\}$ such that $r_{ni}=1$ if $\mathbf{x}_n$ is assigned to the $i$-th cluster, and $r_{ni}=0$ otherwise.
The optimal clustering indicators $\{r_{ni}\}^*$ and quantization levels $\{\mathbf{q}_i\}^*$ are given by
\begin{equation}
\{r_{ni}\}^*,\{\mathbf{q}_i\}^*
=\underset{\{r_{ni}\},\{\mathbf{q}_i\}}{\mathsf{\mathrm{min}}} \sum_{n=1}^N \sum_{i=1}^k r_{ni}\left\Vert \mathbf{x}_n-\mathbf{q}_i \right\Vert^{2}.\label{eqn:kmeans-X}
\end{equation}

Now, let us consider the scaled dataset $\mathbf{Y}=c\mathbf{X}$, in which $\mathbf{y}_{n}$ is the $n$-th point.
In this case, the optimal clustering indicators $\{r_{ni}^\prime\}^*$ and quantization levels $\{\mathbf{q}_i^\prime\}^*$ are
\begin{align}
\{r_{ni}^\prime\}^*,\{\mathbf{q}_i^\prime\}^*
& =\underset{\{r_{ni}^\prime\},\{\mathbf{q}_i^\prime\}}{\mathsf{\mathrm{min}}} \sum_{n=1}^N \sum_{i=1}^k r_{ni}^\prime\left\Vert \mathbf{y}_n-\mathbf{q}_i^\prime \right\Vert^2\\
& =\underset{\{r_{ni}^\prime\},\{\mathbf{q}_i^\prime\}}{\mathsf{\mathrm{min}}} \sum_{n=1}^N \sum_{i=1}^k r_{ni}^\prime\left\Vert c\mathbf{x}_n-\mathbf{q}_i^\prime \right\Vert^2\label{eqn:kmeans-Y1}\\
& =\underset{\{r_{ni}^\prime\},\{\mathbf{q}_i^\prime\}}{\mathsf{\mathrm{min}}} \sum_{n=1}^N \sum_{i=1}^k r_{ni}^\prime\left\Vert \mathbf{x}_n-\frac{1}{c}\mathbf{q}_i^\prime \right\Vert^2\label{eqn:kmeans-Y2},
\end{align}
where the objective function in \eqref{eqn:kmeans-Y1} has been multiplied by the scalar $c^{-2}$ to obtain \eqref{eqn:kmeans-Y2}.
Noting that the problem in \eqref{eqn:kmeans-Y2} is equal to \eqref{eqn:kmeans-X}, it holds $\{r_{ni}^\prime\}^*=\{r_{ni}\}^*$ and $\frac{1}{c}\{\mathbf{q}_i^\prime\}^*=\{\mathbf{q}_i\}^*$.
Thus, we verified that $\{\mathbf{q}_i^\prime\}^*=c\{\mathbf{q}_i\}^*$.

\bibliographystyle{IEEEtran}
\bibliography{IEEEabrv,main}

\end{document}